%% file: MultStopInsurance.tex
\documentclass[a4paper,preprint,10pt,authoryear]{elsarticle_modified}

\usepackage{etex}
 \usepackage[usenames,dvipsnames]{pstricks}
 \usepackage{epsfig}
 \usepackage{pst-grad} 
 \usepackage{pst-plot} 
 \usepackage{caption}
\usepackage{amsfonts,amsmath,amssymb,amscd,amsthm}
\usepackage{latexsym}
\usepackage{multirow}
\usepackage{graphicx}
\usepackage{geometry}
\usepackage{lscape}
\usepackage{bm}
\usepackage{subfig}
\usepackage{url}
\usepackage{pifont}
\usepackage{bbding}
\usepackage{color}
\usepackage{natbib}
\usepackage{float}
\usepackage{placeins}
\usepackage{bm}
\usepackage{algorithm2e}
\usepackage{xcolor}
\usepackage{psbox}

\usepackage[english]{babel}
\usepackage[T1]{fontenc}
\usepackage[ansinew]{inputenc}

\usepackage{lmodern}	

\usepackage{tikz}
\usetikzlibrary{decorations.pathmorphing} 
\usetikzlibrary{fit}					
\usetikzlibrary{backgrounds}	

\providecommand{\U}[1]{\protect\rule{.1in}{.1in}}

\def\R{{\mathbb R}}

\def\E{{\mathbb E}}

\def\one{{\mathbb I}}
\def\F{{\mathcal F}}
\def\P{{\mathbb P}}
\def\F{{\mathcal F}}

\newcommand{\Remm}[1]{}
\newtheorem{thm}{Theorem}[section]
\newtheorem{lemma}[thm]{Lemma}
\newtheorem{prop}[thm]{Proposition}

\newtheorem{defn}[thm]{Definition}
\newtheorem{example}[thm]{Example}
\newtheorem{rem}[thm]{Remarks}

\definecolor {darkblue} {rgb} {0,0.08,0.45}

\begin{document}

\begin{frontmatter}

\title{Optimal insurance purchase strategies via optimal multiple stopping times}
\author{Rodrigo S.~Targino$^{1}$ \quad Gareth W.~Peters$^{1,2}$ \quad Georgy ~Sofronov$^{3}$ \quad Pavel V.~Shevchenko$^{2}$}
\date{{\footnotesize {Working paper, version from \today }}}
\maketitle

\begin{abstract}

\noindent 
In this paper we study a class of insurance products where the policy holder has the option to insure $k$ of its annual Operational Risk losses in a horizon of $T$ years. This involves a choice of $k$ out of $T$ years in which to apply the insurance policy coverage by making claims against losses in the given year. The insurance product structure presented can accommodate any kind of annual mitigation, but we present three basic generic insurance policy structures that can be combined to create more complex types of coverage. Following the Loss Distributional Approach (LDA) with Poisson distributed annual loss frequencies and Inverse-Gaussian loss severities we are able to characterize in closed form analytical expressions for the multiple optimal decision strategy that minimizes the expected Operational Risk loss over the next $T$ years. For the cases where the combination of insurance policies and LDA model does not lead to closed form expressions for the multiple optimal decision rules, we also develop a principled class of closed form approximations to the optimal decision rule. These approximations are developed based on a class of orthogonal Askey polynomial series basis expansion representations of the annual loss compound process distribution and functions of this annual loss.

\vspace{5mm}
\end{abstract}


\begin{center}
{\footnotesize {\ \textit{$^{1}$ Department of Statistical Science,
University College London UCL, London, UK; \\[0pt]
email: gareth.peters@ucl.ac.uk \\[0pt]
(Corresponding Author) \\[0pt]
$^{2}$ CSIRO Mathematics, Informatics and Statistics, Sydney, Australia \\[0pt]
$^{3}$ Macquarie University, Department of Statistics, Sydney, Australia \\[0pt] } } }
\end{center}

\end{frontmatter}
\include{body}

\FloatBarrier

\bibliographystyle{apalike}
\bibliography{IME_MultStop}

\include{appendix}

\end{document}

%% file: body.tex
\section{Introduction} 
\label{section:Intro}
Since the New Basel Capital Accord in 2004, Operational Risk (OpRisk) quantification has become increasingly important for financial institutions. However, the same degree of attention has not yet been devoted to insurance mitigation of OpRisk losses nor, consequently, to detailed analysis of potential risk and capital reduction that different risk transfer strategies in OpRisk may allow. 

Historically the transference of credit and market risks through credit derivatives and interest rate swaps, for example, has been an active subject of extensive studies both from practitioners and academics while only a few references about OpRisk transfer of risk and possible approach to such risk transfers can be found in the literature (see \cite{brandts2004operational}, \cite{bazzarello2006modeling} and \cite{peters2011impact}). This slow uptake of insurance policies in OpRisk for capital mitigation can be partially attributed to four general factors: (a) there still remains a rather limited understanding of the impact on capital reduction of currently available OpRisk insurance products, especially in the complex multi-risk, multi-period scenarios; (b) the relative conservative Basel II regulatory cap of  $20\%$ in a given year (for Advanced Measurement Approach models); (c) the limited understanding at present of the products and types of risk transfer mechanisms available for OpRisk processes; and (d) the limited competition for insurance products available primarily for OpRisk, where yearly premiums and minimum Tier I capital requirements required to even enter into the market for such products precludes the majority of banks and financial institutions in many jurisdictions. 

Some of the reasons for these four factors arises when one realises that OpRisk is particularly challenging to undertake general risk transfer strategies for, since its risk processes range from loss processes which are insurable in a traditional sense (see Definition \ref{DefnInsurableLoss}) to infrequent high consequence loss processes which may be only partially insurable and may result from extreme losses typically covered by catastrophe bonds and other types of risk transfer mechanisms. For these reasons, the development of risk transfer products for OpRisk settings by insurers is a relatively new and growing field in both academic research and industry, where new products are developed as greater understanding of catastrophe and high consequence low frequency loss processes are better understood.

To qualify these points, consider factor (d). In terms of special products, a large reinsurance company that offers a number of products in the space of OpRisk loss processes to a global market is Swiss Re. They have teams such as in the US the Excess and Surplus market Casualty group which specialises in ``U.S-domiciled surplus lines wholesale brokers with primary, umbrella and follow-form excess capacity for difficult-to-place risks in the Excess and Surplus market''. This group aims to seek coverage solutions for challenging risks not in the standard/admitted market. The types of coverage limits offered are quoted as being of the range: USD 10 million limits in umbrella and follow form excess;  USD 5 million CGL limits for each occurrence; USD 5 million general aggregate limit; USD 5 million products/completed operations; and USD 5 million personal and advertising injury. There is also groups like the Professional and Management Liability team in Swiss Re that provide bespoke products for ``protection for organisations and their executives, as well as other professionals, against allegations of wrong-doing, mismanagement, negligence, and other related exposures.'' In addition as discussed in \cite{van2002operational} there are some special products that are available for OpRisk insurance coverage offered by Swiss Re and known as the Financial Institutions OpRisk Insurance (FIORI) which covers OpRisk causes such as Liability, Fidelity and unauthorised activity, Technology risk, Asset protection and External fraud.  It is noted in \cite{chernobai2008operational} that the existence of such specialised products is limited in scope and market since the resulting premium one may be required to pay for such an insurance product can typically run into very significant costs, removing the actual gain from obtaining the insurance contract in terms of capital mitigation in the first place. Hence, although the impact of insurance in OpRisk management is yet to be fully understood it is clear that it is a critical tool for the management of exposures and should be studied more carefully.

At this stage it is beneficial to recall the fundamental definition of an insurance policy or contract.
\begin{defn}[Insurance Policy]
At a fundamental level one can consider insurance to be the fair transfer of risk associated with a loss process between two financial entities. The transfer of risk is formalized in a legal insurance contract which is facilitated by the financial entity taking out the insurance mitigation making a payment to the insurer offering the reduction in risk exposure. The contract or insurance policy legally sets out the terms of the coverage with regard to the conditions and circumstances under which the insured will be financially compensated in the event of a loss. As a consequence the insurance contract policy holder assumes a guaranteed and often known proportionally small loss in the form of a premium payment corresponding to the cost of the contract in return for the legal requirement for the insurer to indemnify the policy holder in the event of a loss.
\end{defn}

Under this definition one can then interpret the notion of insurance as a risk management process in which a financial institution may hedge against potential losses from a given risk process or group of risk processes.  In \cite{mehr1980principles} and \cite{berliner1982limits} they discuss at a high level the fundamental characteristics of what it means to be an insurable loss or risk process, which we note in Definition \ref{DefnInsurableLoss}. We observe that this standard Actuarial view on insurability does not always coincide directly  to the economists view.

\begin{defn}[Insurable Losses] \label{DefnInsurableLoss}
In \cite{mehr1980principles} and \cite[chapter 3]{chernobai2008operational} they define an insurable risk as one that should satisfy the following characteristics:
\begin{enumerate}
\item{The risks must satisfy the ``Law of Large Numbers'', i.e. there should be a large number of similar exposures.}
\item{The loss must take place a known recorded time, place and from a reportable cause.}
\item{The loss process must be considered subject to randomness. That is, the events that result in the generation of a claim should be random or at a minimum outside the control of the policy holder.}
\item{The loss amounts generated by a particular risk process must be commensurate with the charged premium, and associated insurer business costs such as claim analysis, contract issuance and processing.}
\item{The estimated premium associated with a loss process must be affordable.This is particularly important in high consequence rare-event settings, see discussions in \cite{peters2011impact} who consider this question in a general setting.}
\item{The probability of a loss should be able to be estimated for a given risk process as well and some statistic characterizing the typical, average, median etc. loss amount.}
\item{Either the risk process has a very limited chance of a catastrophic loss that would bankrupt the insurer and in addition the events that occur to create a loss occur in a non-clustered fashion; or the insurer will cap the total exposure.}
\end{enumerate}
\end{defn}

In \cite{gollier2005some} they argue that there is also a need to consider the economic ramifications for insurable risks. In particular they add to this definition of insurable risks the need to consider the economic market for such risk transfers. In particular they discuss uninsurable and partially insurable losses, where an uninsurable loss occurs when ``..., given the economic environment, no mutually advantageous risk transfer can be exploited by the consumer and the supplier of insurance''. A partially uninsurable loss arises when the two parties to the risk transfer exchange can only partially benefit or exploit the mutually advantageous components of the risk transfer, this has been considered in numerous studies, see \cite{aase1993equilibrum}, \cite{arrow1953role}, \cite{arrow1965aspects}, \cite{borch1962equilibrium} and \cite{raviv1979design}.

As noted in \cite{gollier2005some}, from the economists perspective the basic model for risk transfer involves a competitive insurance market in which the Law of Large Numbers is utilised as part of the evaluation of the social surplus of the transfer of risk. However, unlike the actuarial view presented above, the maximum potential loss and the probabilities associated with this loss are not directly influential when it comes to assessing the size of risk transfers at market equilibrium. In addition the economic model adds factors related to the degree of risk aversion of market participants (agents) and their degree of optimism when assessing the insurability of risks in the economy. Classically these features are all captured by the economic model know as the Arrow-Borch-Raviv model of perfect competition in insurance markets, see a good review in \cite[section 2]{gollier2005some} and \cite{ghossoub2012belief}.

In the work developed here we pose an interesting general question of how may one construct insurance products satisfying the axioms and definitions above whilst allowing a sufficiently general class of policies that may actually be suitable for a wider range of financial institutions and banks than those specialised products currently on offer. More specifically, in this paper we discuss aspects of an insurance product that provides its owner several opportunities to decide which annual OpRisk loss(es) to insure. This product can be thought of as a way to decrease the cost paid by its owner to the insurance company in a similar way to what occurs with swing options in energy markets (see for example, \cite{jaillet2004valuation} and \cite{carmona2008optimal}): \textit{ instead of buying $T$ yearly insurance policies over a period of $T$ years, the buyer can negotiate with the insurance company a contract that covers only $k$ of the $T$ years (to be chosen by the owner).} This type of structured product will result in a reduction in the cost of insurance or partial insurance for OpRisk losses and this aspect is highlighted in \cite[page 188]{allen2009understanding}, where they note that ``even without considering the cost of major catastrophes, insurance coverage is very expensive''. In addition, we argue it may be interesting to explore such structures if the flexibility they provide results in an increased uptake of such products for OpRisk coverage, further reducing insurance premiums and resulting perhaps in greater competition in the market for these products.

The general insurance product presented here can accommodate any form of insurance policy, but we will focus on three basic generic ``building block'' policies (see Definitions \ref{def:ILP} to \ref{def:PAP}) which can be combined to create more complex types of protection. For these three basic policies we present a ``moderate-tailed'' model for annual risks that leads to closed form usage strategies of the insurance product, answering the question: when is it optimal to ask the insurance company to cover the annual losses? 

For the rest of the paper we assume that throughout a year a financial institution incurs a random number of loss events, say $N$, with severities (loss amounts) $X_1,\hdots,X_N$. Additionally, we suppose the company holds an insurance product that lasts for $T$ years and grants the company the right to mitigate $k$ of its $T$ annual losses through utilisation of its insurance claims. To clarify consider a given year $t \leq T$ where the company will incur $N(t)$ losses adding up to $Z(t) = \sum_{n=1}^{N(t)} X_n(t)$, assuming it has not yet utilised all its $k$ insurance mitigations it then has the choice to make an insurance claim or not. If it utilises the insurance claim in this year the resulting annual loss will be denoted by $\widetilde{Z}(t)$. Such a loss process model structure is standard in OpRisk and insurance and is typically referred to as the Loss Distributional Approach (LDA) which we illustrate an example instance of in Figure \ref{fig:LDAmodel}.

\begin{center}
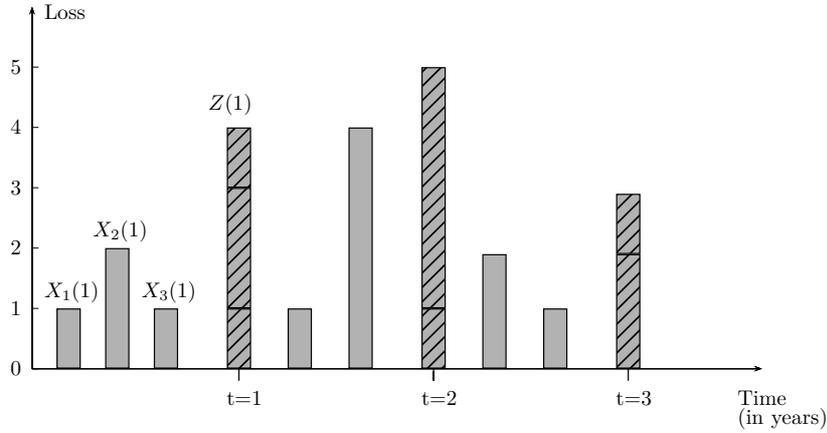

	\input{fig_LDAmodel.tex}
	\captionof{figure}{Schematic representation of a LDA model. The aggregated loss in each year is represented hatched.}
	\label{fig:LDAmodel}
\end{center}

In this context the company's aim is to choose $k$ distinct years out of the $T$ in order to minimize its expected operational loss over the time interval $[0,T]$, where it is worth noting that if $Z>\widetilde{Z}$ ie, if the insurance is actually mitigating the company's losses, all its $k$ rights should be exercised. The question that then must be addressed is what is the optimal decision rule, i.e. define the multiple optimal stopping times for making the $k$ sets of insurance claims.

The rest of the paper is organized as follows. In Section \ref{sec:InsurancePolicies} we present the insurance policies we use as mitigation for the insurance product described above. Section \ref{sec:MultOptDecisionRules} presents an overview of useful theoretical results in the field of multiple stopping rules for independent observations in discrete time, in particular Theorem \ref{thm:MultStop} which is the main result in this Section. A summary of properties related to the LDA model used in this paper is presented in Section \ref{section:LDAandProperties} and used in Section \ref{section:AnalyticalOptimalRules} to present the main contribution of this work, namely closed form solutions for the optimal multiple stopping rules for the insurance products considered. In Section \ref{sec:CaseStudies} we check the theoretical optimality of the rules derived in Section \ref{section:AnalyticalOptimalRules}, comparing them with predefined rules.

Since these closed form results rely upon the stochastic loss model considered, we also provide a general framework applicable for any loss process. Therefore, in Section \ref{sec:SeriesExpansion} we discuss a method based on series expansions of unknown densities to calculate the optimal rules when the combination of insurance policy and severity density does not lead to analytical results. The conclusions and some final considerations are shown in Section \ref{sec:conclusion}.

\section{Insurance Policies} 
\label{sec:InsurancePolicies}
As stated before, the insurance policies presented here must be thought as building blocks for more elaborated ones, leading to mitigation of more complex sources of risk. It also worth noticing that the policies presented are just a mathematical model of the actual policies that would be sold in practice and although some characteristics, such as deductibles, can be incorporated in the model they are not presented at this stage.

In the sequel we present these basic insurance policies the company can use in the insurance product. For the sake of notational simplicity, if a process $\big\{Z(t) \big\}_{t=1}^T$ is a sequence of i.i.d. random variables, we will drop the time index and denote a generic r.v. from this process as $Z$. For the rest of the paper $\one_A$ will denote the indicator function on the event $A$, ie, $\one_A= 1$ if $A$ is valid and zero otherwise.

\begin{defn}[Individual Loss Policy (ILP)] \label{def:ILP} This policy applies a constant haircut to the loss process in year $t$ in which individual losses experience a Top Cover Limit (TCL) as specified by
			\begin{equation*}
				\widetilde{Z} = \sum_{n=1}^{N} \max\left(X_{n} - \text{TCL}, \ 0\right).
			\end{equation*}
\end{defn}

\begin{defn}[Accumulated Loss Policy (ALP)] \label{def:ALP} The ALP provides a specified maximum compensation on losses experienced over a year. If this maximum compensation is denoted by $ALP$ then the annual insured process is defined as
	\begin{equation*}
		\widetilde{Z} = \left(\sum_{n=1}^{N} X_n - ALP \right) \one_{\left\{\sum_{n=1}^{N} X_n > ALP\right\}}.
	\end{equation*}
\end{defn}

\begin{defn}[Post Attachment Point Coverage (PAP)] \label{def:PAP} The attachment point is the insured's retention point after which the insurer starts compensating the company for accumulated losses at point $PAP$
			$$\widetilde{Z} = \sum_{n=1}^{N} X_{n} \times \one_{\left\{ \sum_{k=1}^{n} X_k \leq PAP \right\}}. $$
\end{defn}

To characterize the annual application of such policies we provide a schematic representation of each of these policies in Figures \ref{fig:TCL} to \ref{fig:PAP}, assuming the same losses as in Figure \ref{fig:LDAmodel}. The (part of the) loss mitigated by the insurance policy is represented by a white bar and the remaining loss due to the owner of the insurance product is painted grey. As in Figure \ref{fig:LDAmodel}, annual losses are represented by hatched bars.

\begin{center}
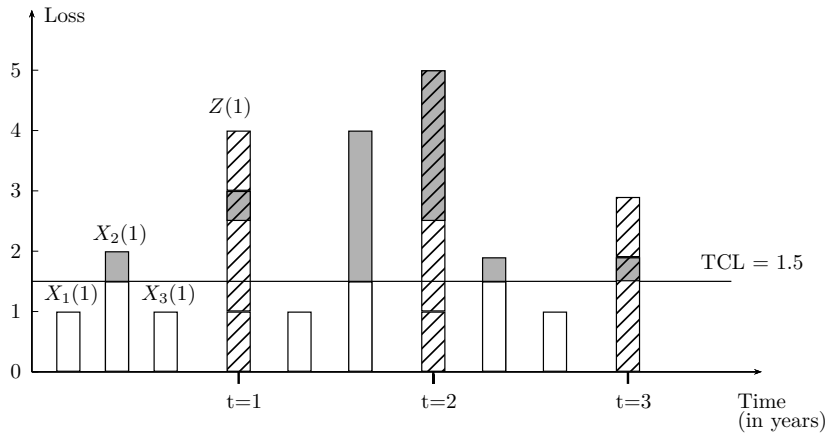

	\input{fig_TCL.tex}
	\captionof{figure}{Individual Loss Policy (ILP) with TCL level of 1.5.} 
	\label{fig:TCL}
\end{center}

\begin{center}
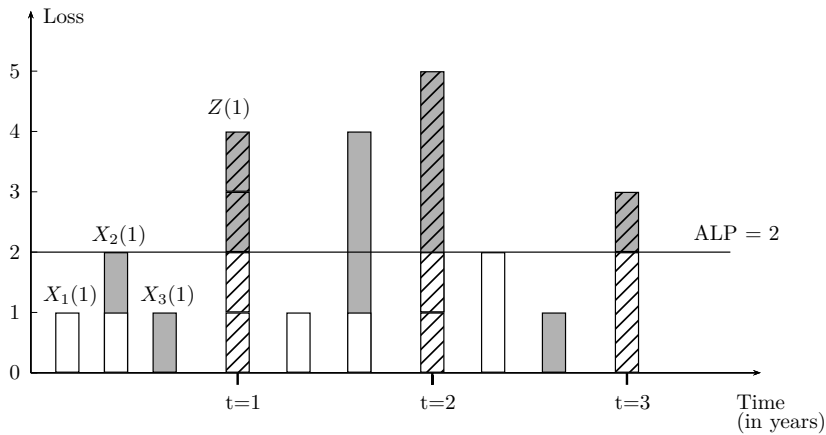

	\input{fig_ALP.tex}
	\captionof{figure}{Accumulated Loss Policy (ALP) with ALP level of 2.0.}
	\label{fig:ALP}
\end{center}

\begin{center}
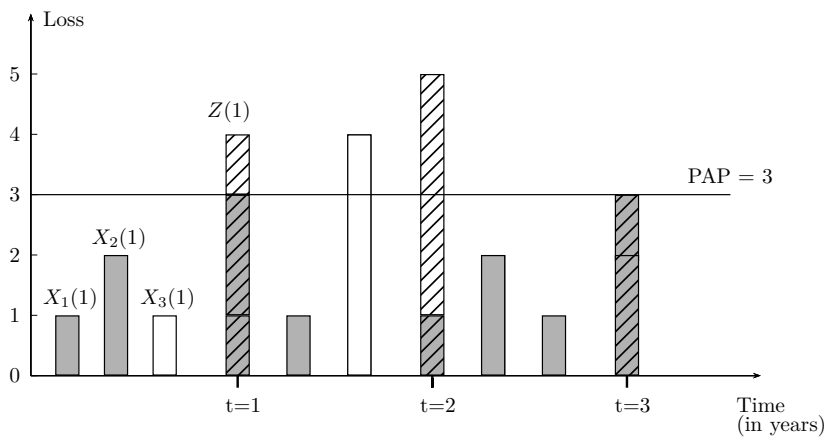

	\input{fig_PAP.tex}
	\captionof{figure}{Post Attachment Point Coverage policy (PAP) with PAP level of 3.0}
	\label{fig:PAP}
\end{center}

\section{Multiple Optimal Decision Rules} 
\label{sec:MultOptDecisionRules}
Assume an agent sequentially observe a process $\big\{W(t)\big\}_{t=1}^{T}$, for a fixed $T< +\infty$ and wants to choose $k<T$ of these observations in order to maximize (or minimize, see Remark \ref{minimizationRemark}) the expected sum of these chosen observations. For $k=1$, this problem is known in the literature as the house selling problem (see \cite{sofronov2012optimal} for an updated literature review) since one of its interpretations is as follows. If the agent is willing to sell a house and assume that at most $T$ bids will be observed he wants to choose the optimal time $\tau$ such that the house will be sold for the highest possible value. The extension of this problem for $k>1$ is know as the multiple house selling problem, where the agent wants to sell $k$ identical houses. It is worth noting that in our insurance problem the agent is interested in choosing $k$ periods to exercise the insurance policy in order to minimize loss, in a sense that will be make precise shortly in this paper.

Formally, the mathematical framework of this problem consists of a filtered probability space $\big(\Omega, \F, \{\F_t\}_{t \geq 0}, \P\big)$, where $\F_t=\sigma\big(W(t)\big)$ is the sigma-algebra generated by $W(t)$. Within this framework, where we assume the flow of information is given only by the observed values of $W$, it is clear that any decision at time $t$ should take into account only values of the process $W$ up to time $t$. It is also required that two actions can not take place at the same time, ie, we do not allow two stopping times to occur at the same discrete time instant. These assumptions are precisely stated in the following definition, but for further details on the theory of multiple optimal stopping rules we refer the reader to \cite{nikolaev2007multiple} and \cite{sofronov2012optimal}.

\begin{defn} A collection of integer-valued random variables $(\tau_1,\hdots,\tau_i)$ is called an $i$-multiple stopping rule if the following conditions hold:
\begin{enumerate}[(a)]
	\item $\{\omega \in \Omega \ : \ \tau_1(\omega)=m_1,\hdots,\tau_j(\omega)=m_j\} \in \F_{m_j}, \ \forall m_j>m_{j-1}>\hdots>m_1\geq 1, \ j=1,\hdots,i$;
	\item $1 \leq \tau_1<\tau_2<\hdots<\tau_i<+\infty, \ (\P\text{-a.s.})$.
\end{enumerate}
\end{defn}

Given the mathematical definition of a stopping rule the notion of optimality of these rules can be made precise in the following definitions.

\begin{defn}
For a given multiple stopping rule $\boldsymbol{\tau}=(\tau_1,\hdots,\tau_k)$ the gain function utilized in this paper takes the following additive form:
$$g(\boldsymbol{\tau}) = W(\tau_1)+\hdots+W(\tau_k).$$
\end{defn}

\begin{defn} Let $\mathcal{S}_m$ be the class of multiple stopping rules $\boldsymbol\tau=(\tau_1,\hdots,\tau_k)$ such that $\tau_1\geq m \ (\P\text{-a.s.}).$ The function
$$v_m = \sup_{\boldsymbol\tau \in \mathcal{S}_m} \E[g(\boldsymbol\tau)]$$
is defined as the m-value of the game and, in particular, if $m=1$ then $v_1$ is the value of the game.
\end{defn}

\begin{defn} A multiple stopping rule $\boldsymbol\tau^* \in \mathcal{S}_m$ is called an optimal multiple stopping rule in $\mathcal{S}_m$ if $\E[W(\boldsymbol\tau^*)]$ exists and $\E[W(\boldsymbol\tau^*)]=v_m$.
\end{defn}

The following result (first presented in \cite{nikolaev2007multiple}, Theorem 3) provides the optimal multiple stopping rule that maximizes the expectation of the sum of the observations.

\begin{thm} \label{thm:MultStop}Let $W(1),W(2),\dots,W(T)$ be a sequence of independent random variables with known distribution
functions $F_{1},F_{2},\dots,F_{T}$, and the gain function $g(\boldsymbol{\tau}) =\sum_{j=1}^k W(\tau_j)$. Let $v^{L,l}$ be the value of a game where the agent is allowed to stop $l$ times $(l\leqslant k)$ and there are $L$ $(L \leqslant T)$ steps remaining. If there exist $\E[W(1)],\E[W(2)],\dots,\E[W(T)]$ then the value of the game is given by
\begin{align*}
	v^{1,1} &= \E[W(T)], \\
	v^{L,1} &= \E\bigl[\max\{W(T-L+1), v^{L-1,1} \}\bigr], \ \  1 < L \leq T, \\
	v^{L,l+1} &= \E\bigl[\max\{v^{L-1,l}+W(T-L+1),v^{L-1,l+1}\}\bigr], \ \ l+1 < L \leq T, \\
	v^{l,l} &= \E\left[v^{l-1,l-1} + W(T-l+1)\right].
\end{align*}

If we put
\begin{equation}
\begin{array}{r@{}l} \label{eq:tau}
    \tau^*_1 &{}= \min\{ m_1 \, : \, 1 \leqslant m_1 \leqslant T-k+1, W(m_1) \geqslant v^{T-m_1,k}-v^{T-m_1,k-1}\}; \\
    \tau^*_i &{}= \min\{ m_i \, : \, \tau^*_{i-1} < m_i \leqslant T-k+i, W(m_i) \geqslant v^{T-m_i,k-i+1}-v^{T-m_i,k-i}\}, \ i=2,\dots,k-1; \\
		\tau^*_k &{}= \min\{ m_k \, : \,\tau^*_{k-1} < m_k \leqslant T, W(m_k) \geqslant
v^{T-m_k,1}\};
\end{array}
\end{equation}
then $\boldsymbol\tau^*=(\tau^*_1,\dots,\tau^*_k)$ is the optimal multiple stopping rule. 
\end{thm}

In the context we consider it will always be optimal to stop the process exactly $k$ times, but this may not be true, for example, if some reward is given to the product holder for less than $k$ years of claims of insurance. In the absence of such considerations, we proceed with assuming always $k$ years of claims will be made. In Theorem \ref{thm:MultStop} we can see that the value function for $L>l$ is artificial and $v^{0,1}$, for example, has no interpretation. On the other hand, $v^{1,1}$ can not be calculated using the general formula (it would depend on $v^{0,1}$). With one stop remaining and one step left, from the reasons given above, we are obliged to stop, and, therefore, there is no maximization step when calculating $v^{1,1}$, i.e., $v^{1,1}= \E[W(T-1+1)]$.
The same argument is valid for $l>1$ and, in this case,
$$v^{L,l} = \E\left[\max\{v^{L-1, l-1} + W(T-L+1), \ v^{L-1,l} \}\right], \ 1 \leq l \leq T$$
and, if we have $l \leq (T-1)$ steps left and also $l$ stops, we must stop in all the steps remaining. So,
$$v^{l,l} = \E\left[v^{l-1,l-1} + W(T-l+1)\right].$$

%
From Theorem \ref{thm:MultStop} and the assumption of independence of the annual losses, we can see that to be able to calculate the optimal rule we only need to calculate (unconditional) expectations like $\E[W]$ and $\E[\max\{c_1+W, \ c_2 \}]$, for different values of $c_1$ and $c_2$. In addition, since $0 \leq v^{L-1,l} \leq v^{L-1,l+1}$, we actually only need to calculate $\E[\max\{c_1+W, \ c_2 \}]$ for $0\leq c_1 \leq c_2$.

\begin{figure}[h]
\centering
		\includegraphics[scale = 0.2]{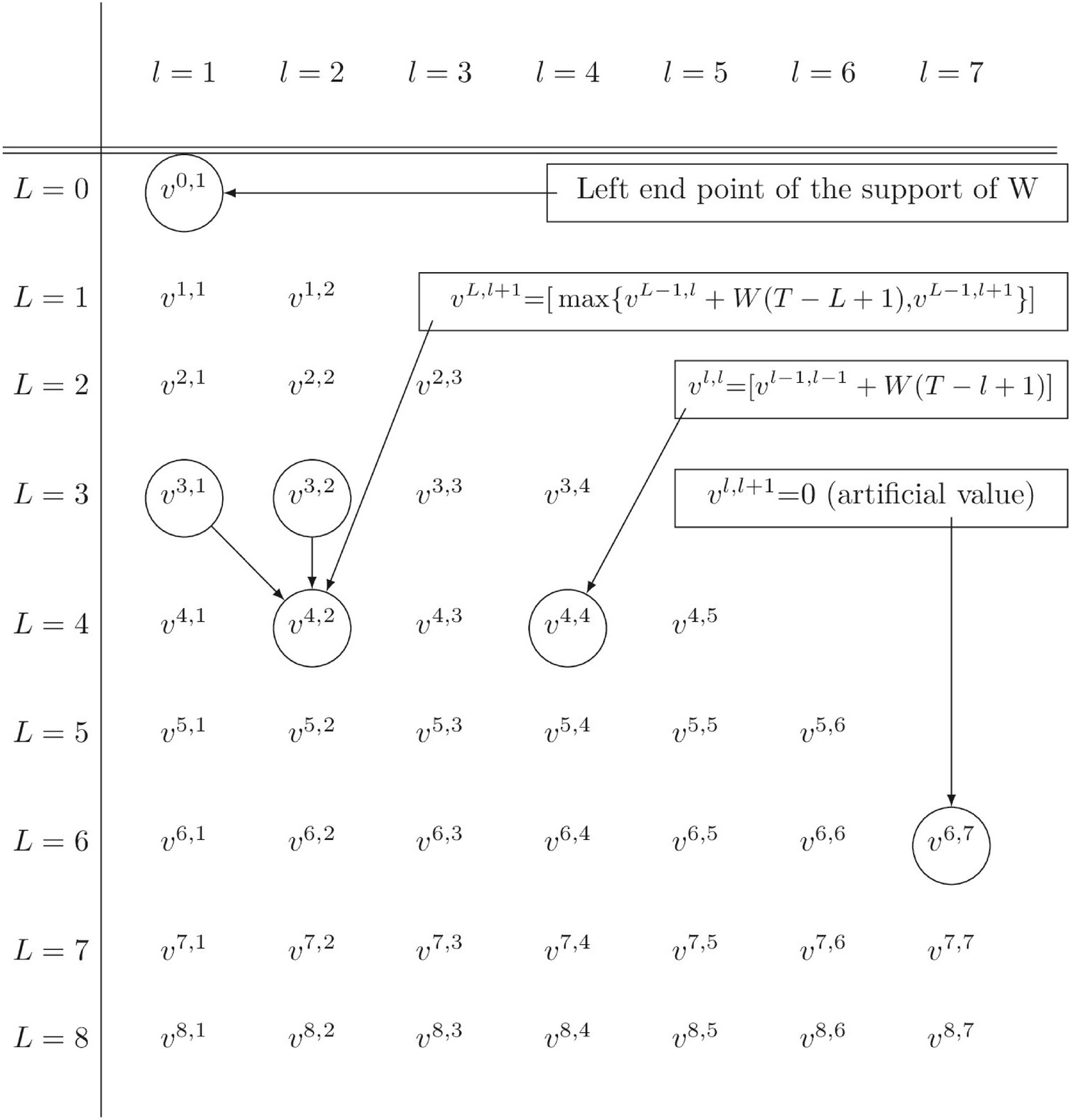}
		\label{fig:valueFunction}
		\caption{Schematic representation of the value function iteration.}
\end{figure}

\subsection{Objective Functions for Rational and Boundedly Rational Insurees}
In this paper we will consider two possible general populations for the potential insuree. The first group are those that are perfectly rational, meaning that they will always act in an optimal fashion when given the chance and, more importantly, are capable (i.e. have the resources) to figure out what is the optimal behaviour. In this case we will consider a global objective function to be optimized.

The second group represent boundedly rational insurees who act sub-optimally. This group represents firms who are incapable or lack the resources/knowledge to understand how to act optimally when determining their optimal behaviours/actions and will be captured by local behaviours.

Hence, these two populations will be encoded in two objective functions: one which is optimal (globally) and one which represents a sub-optimal (local strategy) the boundedly rational population would likely adopt. These behaviours can be made precise through the following exercising strategies, for the first and second groups, respectively.
\begin{enumerate}
	\item \textbf{Global Risk Transfer Strategy:} Minimizes the (expected) total loss over the period $[0,T]$;
	\item \textbf{Local Risk Transfer Strategy:} Minimizes the (expected) sum of the losses at the insurance times (i.e. stopping times).
\end{enumerate}

These two different groups can be understood as, for example, large corporations, with employees dedicated to fully understand the mathematical nuances of this kind of contract and small companies, with limited access to information. The group with ``bounded rationality'' may decide (heuristically, without the usage of any mathematical tool) to follow the so-called Local Risk Transfer Strategy, which will produce smaller gain in the period $[0,T]$. As we will see in Section \ref{sec:CaseStudies} these two different objective functions can lead to completely different exercising strategies, and we believe the insurance company who sells this contract should be aware of these different behaviours.

For the first loss function the formal objective is to minimize
$$\sum_{t=1 \atop t\notin \{\tau_1,\hdots,\tau_k \}}^T \hspace{-0.5cm} Z(t) + \sum_{j=1}^k \widetilde{Z}(\tau_j) = \sum_{t=1}^T Z(t) - \hspace{-0.5cm} \sum_{t=1 \atop t\in \{\tau_1,\hdots,\tau_k \}}^T \hspace{-0.5cm} \Big\{  Z(t) - \widetilde{Z}(t)\Big\}. $$
Since $\sum_{t=1}^T Z(t)$ does not depend on the choice of $\tau_1,\hdots,\tau_k$, this is, in fact, equivalent to maximize
$$\sum_{j=1}^k W(\tau_j) =  \sum_{j=1}^k \Big\{  Z(\tau_j) - \widetilde{Z}(\tau_j)\Big\},$$
where the process $W$ is defined as $W(t) = Z(t) - \widetilde{Z}(t)$.

For the second objective function, the company aims to minimise the total loss not over period $[0,T]$ but instead only at times at which the decisions are taken to apply insurance and therefore claim against losses in the given year,
$$\sum_{j=1}^k \widetilde{Z}(\tau_j)$$
and, in this case, the process $W$ should be viewed as $W(t) = - \widetilde{Z}(t)$.

\begin{rem} Note that if the agent is trying to maximize the first loss function (using $W = Z - \widetilde{Z}$), then $W$ is non-negative stochastic process, and only one kind of expectation is required to be calculated, since if $c_1=c_2=0$, then $\E[\max\{c_1+W, \ c_2 \}] = \E[W]$.  
\end{rem}

\begin{rem}\label{minimizationRemark} If the agent is trying to minimize the expected gain of the sum of $\widetilde{Z}(t)$ random variables (instead of maximizing it) one can rewrite the problem as follows. Define a process $W(t) = -\widetilde{Z}(t)$ and note that $\min \E[ \sum_{j=1}^k\widetilde{Z}(\tau_j)] = \max \E[ \sum_{j=1}^k W(\tau_j)]$. Therefore the optimal stopping times that maximize the expected sum of the process $W$ are the same that minimize the expected sum of the process $\widetilde{Z}$.
\end{rem}

Although this work is mainly devoted to study of a combination of insurance policy and severity distribution that leads to closed form results of the value function integrals required for closed form multiple optimal stopping rules, we also show how one can develop principled approximation procedures can be used to calculate the distribution of the insured process $\widetilde{Z}$ (see Section \ref{sec:SeriesExpansion}). In the remainder of this section we present a very simple example using the second (local) objective function, where we assume the annual insured losses are modelled as Log-Normal random variables.

\begin{example}[Log-Normal] Assume that the insured losses $\widetilde{Z}(1),\hdots,\widetilde{Z}(T)$ form a sequence of i.i.d. random variables such that $\widetilde{Z} \sim$ Log-Normal(0,1). To calculate the multiple optimal rule that minimizes the expected loss let us define $W=-\widetilde{Z}$. The values of the game using the equations in Theorem \ref{thm:MultStop} can be seen in Table \ref{LN_example}.

Note that Table \ref{LN_example} presents the value of expected loss \underline{\emph{at the times we stop}}, ie, $\E\left[\sum_{j=1}^k - \widetilde{Z}(\tau_j)\right],$ so it only makes sense to compare values within the same column. Doing so one can see that for a fixed number of stops $l$, the value of the game is increasing with the number of steps remaining. In other words the more one can wait to decide in which step to stop the smaller is the expected loss.

\begin{table}[h]
\centering
\begin{tabular}{rrrrrrrrrr}
  \hline
	L	& l=1 & l=2 & l=3 & l=4 & l=5 & l=6 & l=7 & l=8 & l=9 \\ 
  \hline
	0 & 0.00 &  &  &  &  &  &  &  &  \\ 
  1& -1.65 & 0.00 &  &  &  &  &  &  &  \\ 
  2 & -1.02 & -3.30 & 0.00 &  &  &  &  &  &  \\ 
  3 & -0.77 & -2.19 & -4.95 & 0.00 &  &  &  &  &  \\ 
  4 & -0.64 & -1.71 & -3.45 & -6.59 & 0.00 &  &  &  &  \\ 
  5 & -0.55 & -1.43 & -2.76 & -4.77 & -8.24 & 0.00 &  &  &  \\ 
  6 & -0.49 & -1.25 & -2.34 & -3.87 & -6.12 & -9.89 & 0.00 &  &  \\ 
  7 & -0.44 & -1.12 & -2.05 & -3.32 & -5.04 & -7.51 & -11.54 & 0.00 &  \\ 
  8 & -0.41 & -1.02 & -1.85 & -2.94 & -4.36 & -6.26 & -8.91 & -13.19 & 0.00 \\ 
  9 & -0.38 & -0.94 & -1.69 & -2.65 & -3.88 & -5.45 & -7.50 & -10.34 & -14.84 \\ 
  10 & -0.36 & -0.88 & -1.56 & -2.43 & -3.52 & -4.87 & -6.58 & -8.78 & -11.78 \\ 
   \hline
\end{tabular}
	\caption{Table of the value function for different $L$ (steps remaining) and $l$ stops in the Log-Normal example.}\label{LN_example}
\end{table}

If we suppose that $T=7$, and we are granted four stops the expected loss is $v^{7,4}=-3.32$. In this case the optimal stopping rule is given by
\begin{align*}
	\tau_1^* &= \min\{ m_1 \, : \, 1 \leqslant m_1 \leqslant 4, W(m_1) \geqslant v^{7-m_1,4}-v^{7-m_1,3}\} \\
	\tau_2^* &= \min\{ m_2 \, : \, \tau_1^* < m_2 \leqslant 5, W(m_2) \geqslant v^{7-m_2,3}-v^{7-m_2,2}\} \\
	\tau_3^* &= \min\{ m_3 \, : \, \tau_2^* < m_3 \leqslant 6, W(m_3) \geqslant v^{7-m_3,2}-v^{7-m_3,1}\} \\
	\tau_4^* &= \min\{ m_4 \, : \, \tau_3^* < m_4 \leqslant 7, W(m_4) \geqslant v^{7-m_4,1}\}.
\end{align*}
For instance, if we observe the sequence $$w_1=-0.57, w_2=-0.79, w_3=-4.75, w_4=-1.07, w_5= -1.14, w_6= -5.56, w_7=-1.59$$ then the optimal stopping times are given by:
\begin{align*}
	\tau_1^* &= 1, \text{ since } w_1=-0.57 \geq v^{7-1,4}-v^{7-1,3} = -3.87 - (-2.34) = -1.53; \\
	\tau_2^* &= 2, \text{ since } w_2=-0.79 \geq v^{7-2,3}-v^{7-2,2} = -2.76 - (-1.43) = -1.33; \\
	\tau_3^* &= 4, \text{ since } w_4=-1.07 \geq v^{7-4,2}-v^{7-4,1} = -2.19 - (-0.77) = -1.42; \\
	\tau_4^* &= 7, \text{ because we are obliged to stop exactly 4 times}.
\end{align*}
In this case the realized loss \underline{\emph{at the stopping times}} is $-0.57 -0.79 -1.07 -1.59 = -4.02$, wich should be comparable with the expected loss under the optimal rule: $-3.32$.
\end{example}

\section{Loss Process Models via Loss Distributional Approach} 
\label{section:LDAandProperties}
Before discussing the application of the Theorem \ref{thm:MultStop} to the problem of choosing the multiple exercising dates of the insurance product present in Section \ref{section:Intro}, in this Section we present the LDA model that leads to closed form solutions in Section \ref{section:AnalyticalOptimalRules}.

The Loss Distributional Approach (LDA) in OpRisk assumes that during a year $t$ a company suffers $N(t)$ operational losses, with $N(t)$ following some counting distribution (usually Poisson or Negative Binomial). The severity of each of these losses is denoted by $X_1(t),\hdots,X_{N(t)}(t)$ and the cumulative loss by the end of year $t$ is given by $Z(t) = \sum_{n=1}^{N(t)} X_n(t)$. For the purpose of modelling OpRisk losses it is essential that the severity density allows extreme events to occur, since these events often occur in practice, as shown, for example, in \cite[Section 1.1]{peters2013understanding}. Following the nomenclature in \cite[Table 3.3]{franzetti2011operational}, the Inverse Gaussian distribution possess a ``moderate tail'' which makes it a reasonable model for OpRisk losses for many risk process types and is often used in practice. This family of distributions also has the advantage of being closed under convolution and this characteristic is essential if closed form solutions for the multiple optimal stopping problem are to be obtained.

In the closed form solutions we present for the different insurance policies we use properties of the Inverse Gaussian distribution and its relationship with the Generalized Inverse Gaussian distribution. The following Lemmas will be used throughout; see additional details in \cite{folks1978inverse} and \cite{jorgensen1982gig}.

In the following, let $X_1,\hdots,X_n$ be a sequence of i.i.d. Inverse Gaussian (IG) random variable with parameters $\mu, \lambda > 0$, ie,
$$f_{X}(x; \ \mu,\lambda) = \left( \frac{\lambda}{2\pi}\right)^{1/2} x^{-3/2} \exp\left\{ \frac{-\lambda(x - \mu)^2}{2\mu^2x}\right\}, \ \ x>0.$$
Let also G be a Generalized Inverse Gaussian (GIG) r.v. with parameters $\alpha, \beta >0$, $p \in \R$, ie,
$$f_{G}(x; \ \alpha,\beta,p) = \frac{(\alpha/\beta)^{p/2}}{2K_p(\sqrt{\alpha \beta})}x^{p-1}\exp\left\{ -\frac{1}{2} (\alpha x + \beta/x)\right\}, \ \ x>0,$$
where $K_p$ is a modified Bessel function of the third kind (sometimes called modified Bessel function of the second kind), defined as
$$K_p(z) = \frac{1}{2} \int_0^{+\infty} u^{p-1}e^{-z(u+1/u)/2}du.$$

\begin{lemma} \label{property1} The Inverse Gaussian family of random variables is closed under convolution and the distribution of its sum is given by
\begin{equation}
	S_n := \sum_{l=1}^n X_l \sim IG(n\mu, \ n^2 \lambda).
\end{equation}
\end{lemma}

\begin{lemma} \label{property2} Any Inverse Gaussian random variable can be represented as a Generalized Inverse Gaussian, and for the particular case of Lemma \ref{property1} the relationship is
\begin{equation}
	f_{S_n}(x;\ n\mu, n^2\lambda) \equiv f_G(x; \ \lambda/\mu^2, n^2\lambda, -1/2).
\end{equation}
\end{lemma}

\begin{lemma} \label{property3} Modified Bessel functions of the third kind are symmetric around zero in the parameter $p$. In particular when $p=1/2$,
\begin{equation}
	\frac{K_{1/2}(\frac{n \lambda}{\mu})}{K_{-1/2}(\frac{n \lambda}{\mu})} = 1.
\end{equation}
\end{lemma}

\begin{lemma} \label{property4} The density of an Inverse Gaussian r.v. has the following property (which clearly holds for any power of $x$, with the proper adjustment in the last parameter of the GIG in the RHS):
\begin{equation}
	x f_G(x; \ \lambda/\mu^2, n^2\lambda, -1/2) \equiv n\mu \ f_G(x; \ \lambda/\mu^2, n^2\lambda, 1/2). 
\end{equation}
\end{lemma}

\begin{proof}(of Lemmas \ref{property1}--\ref{property4})
The proof of Lemma \ref{property1} can be found in \cite[Section 2]{tweedie1957statistical} and the result in Lemma \ref{property2} can be seen by comparing the kernel of both distributions.

The symmetry in Lemma \ref{property3} can be seen through the following characterization of modified Bessel functions of the third kind
$$K_p(x) := \int_0^{+\infty} \exp\left\{-x \cosh(t) \right\} \cosh(p t) dt,$$
(see \cite{watson1922bessel}, page 181) and the fact that $\cosh(-p) = (-1)\cosh(p).$
The last result, Lemma \ref{property4}, follows from Lemma \ref{property3} and a simple comparison of the densities.
\end{proof}

\section{Closed-Form Multiple Optimal Stopping Rules for Multiple Insurance Purchase Decisions} \label{section:AnalyticalOptimalRules}
In this Section we present some models in which the optimal rules can be calculated explicitly, with all the technical proofs postponed to the Appendix. Using the results presented in Section \ref{section:LDAandProperties} we show that if we assume a Poisson-Inverse Gaussian LDA model, where $X_n\sim IG(\lambda, \mu)$ and $N\sim Poi(\lambda_N)$, the optimal times (years) to exercise or make claims on the insurance policy for the Accumulated Loss Policy (ALP) and the Post Attachment Point Coverage Policy (PAP) can be calculated analytically regardless of where the global or local gain (objective) functions are considered. For the Individual Loss Policy (ILP), when using the gain function as the local objective function given by the sum of the losses at the stopping times (insurance claim years) we propose to model the losses after the insurance policy is applied and, in this case, we present analytical solutions for the stopping rules. On the other hand, the ILP Total loss case given by the global objective function does not produce a closed form solution. However, we show how a simple Monte Carlo scheme can be used to accurately estimate the results.

Since we assume the annual losses $Z(1),\hdots,Z(T)$ are identically distributed we will denote by $Z$ a r.v. such that $Z \sim Z(1)$. As in the other Sections $\widetilde{Z}$ is the insured process; $S_n =\sum_{k=1}^n X_k$ is the partial sum up to the n-th loss; $p_m= \P[N=m]$ is the probability of observing $m$ losses in one year. The gain $W$ will be defined as either $-\widetilde{Z}$, when the objective is to minimize the loss at the times the company uses the insurance policy (local optimality), or $Z - \widetilde{Z}$, in case the function to be minimized is the total loss over the time horizon $[0,T]$, i.e. (global optimality).

\subsection{Accumulated Loss Policy (ALP)}

For the ALP case (see Definition \ref{def:ALP}) we can model the severity of the losses before applying the insurance policy. Conditional upon the fact that $\sum_{n=1}^m X_n > ALP$, then the annual loss after the application of the insurance policy will be $\sum_{n=1}^m X_n - ALP$. With this in mind, we can then calculate the c.d.f.'s of the insured process, $\widetilde{Z}$ and also of the random variable $Z-\widetilde{Z}$.

\subsubsection{Local Risk Transfer Objective: Minimizing the Loss at the Stopping Times}
\begin{prop}[Local Risk Transfer Case] \label{prop:ALP} The cdf and pdf of the insured process are given, respectively, by
\begin{align}
	F_{\widetilde{Z}}(z) &=\sum_{m=1}^{+\infty} F_{IG}(z+ALP; \ m\mu, m^2\lambda) C_m + C_0, \label{cdfALP} \\
	f_{\widetilde{Z}}(z) &= \sum_{m=1}^{+\infty} \Big\{f_{IG}(z+ALP; \ m\mu, m^2\lambda) C_m \Big\}\one_{\{z>0\}} + C_0\one_{\{z=0\}} \label{pdfALP};
\end{align}
where the constants $C_0, C_1,C_2,\hdots$ are defined as
\begin{align*}
	C_0&:=\sum_{m=1}^{+\infty}F_{IG}(ALP; \ m\mu, m^2\lambda)p_m  + p_0 \\
	C_m&:=\overline{F}_{IG}(ALP; \ m\mu, m^2\lambda) p_m, \ m=1,2,\hdots \\
\end{align*}
\end{prop} 

After calculating the distribution of $\widetilde{Z}$ we can calculate expectations of the form $\mathbb{E}\left[\max\left\{c_1 + W, c_2\right\} \right]$ w.r.t. the loss process $Z$ and, therefore one can consequently obtain the multiple optimal stopping rules under the Accumulated Loss Policy via direct application of Theorem \ref{thm:MultStop}.

\begin{thm}[Local Risk Transfer Case] \label{thm:ALP} Using the notation of Theorem \ref{thm:MultStop} and defining $W(t) = -\widetilde{Z}(t)$, for $t=1,\hdots,T$ the multiple optimal stopping rule is given by the set of equations in (\ref{eq:tau}), where
\begin{align*}
	v^{1,1} &= -\sum_{m=1}^{+\infty} C_m \Big(m\mu \overline{F}_{GIG}(ALP; \ \lambda/\mu^2, m^2\lambda, 1/2) - ALP \overline{F}_{GIG}(ALP; \ \lambda/\mu^2, m^2\lambda, -1/2) \Big), \\
	v^{L,1} &= -\sum_{m=1}^{+\infty} C_m \Bigg[\Bigg(m\mu \Big(F_{GIG}(v^{L-1,1}+ALP; \ \lambda/\mu^2, m^2\lambda, 1/2)- F_{GIG}(ALP; \ \lambda/\mu^2, m^2\lambda, 1/2)\Big) \\
				  &+ ALP\Big(F_{GIG}(v^{L-1,1}+ALP; \ \lambda/\mu^2, m^2\lambda, -1/2)- F_{GIG}(ALP; \ \lambda/\mu^2, m^2\lambda, -1/2) \Big) \Bigg) \\
					&+ v^{L-1,1} \overline{F}_{GIG}(v^{L-1,1}+ALP; \ \lambda/\mu^2, m^2\lambda, -1/2)\Bigg],\\
	v^{L,l+1} &= -\sum_{m=1}^{+\infty} C_m \Bigg[\Bigg(m\mu \Big(F_{GIG}(v^{L-1,l+1}-v^{L-1,l}+ALP; \ \lambda/\mu^2, m^2\lambda, 1/2)- F_{GIG}(ALP; \ \lambda/\mu^2, m^2\lambda, 1/2)\Big) \\
						&- (v^{L-1,l}-ALP) \Big(F_{GIG}(v^{L-1,l+1}-v^{L-1,l}+ALP; \ \lambda/\mu^2, m^2\lambda, -1/2)- F_{GIG}(ALP; \ \lambda/\mu^2, m^2\lambda, -1/2) \Big) \Bigg) \\
						&+ v^{L-1,l+1} \overline{F}_{GIG}(v^{L-1,l+1}-v^{L-1,l}+ALP; \ \lambda/\mu^2, m^2\lambda, -1/2)\Bigg] - v^{L-1,l}C_0,\\
	v^{l,l} &= v^{l-1,l-1} -\sum_{m=1}^{+\infty} C_m \Big(m\mu \overline{F}_{GIG}(ALP; \ \lambda/\mu^2, m^2\lambda, 1/2) - ALP \overline{F}_{GIG}(ALP; \ \lambda/\mu^2, m^2\lambda, -1/2) \Big).
\end{align*}
\end{thm}

\subsubsection{Global Risk Transfer Objective: Minimizing the Loss Over Period $[0,T]$}
If we assume the company wants to minimize its total loss over the period $[0,T]$ the gain achieved through the Accumulated Loss Policy (ALP) is given by
\begin{align*}
W &= Z - \widetilde{Z} \\
  &= \sum_{n=1}^N X_n -  \left(\sum_{n=1}^N X_n - ALP \right) \one_{\left\{\sum_{n=1}^N X_n > ALP\right\}} \\
	&= ALP \one_{\left\{\sum_{n=1}^N X_n > ALP\right\}} +  \left(\sum_{n=1}^N X_n \right) \one_{\left\{\sum_{n=1}^N X_n > ALP\right\}} \\
	&= \min\left\{ALP, \ \sum_{n=1}^N X_n \right\}.
\end{align*}
For notational convenience we will denote by $W_m = \min\left\{ALP, \ \sum_{n=1}^m X_n \right\}$ the annual gain conditional on the fact that $m$ losses were observed.

\begin{prop}[Global Risk Transfer Case: ALP] \label{prop:ALP_new} The cdf and pdf of the gain process are given, respectively, by
\begin{align}
	F_W(w) &=\one_{\left\{w \geq ALP \right\}} + F_{S_m} (w)\one_{\left\{w < ALP \right\}}, \\
	f_W(w) &= \sum_{m=1}^N\Big\{ \Big( \overline{F}_{S_m}(ALP)\one_{\left\{w = ALP \right\}} + f_{S_m} (w)\one_{\left\{0< w < ALP \right\}}   \Big)p_m \Big\}+ p_0\one_{\left\{w =0 \right\}} \label{pdfALP_new}.
\end{align}
\end{prop} 

After calculating the distribution of the gain, $W$, we can calculate expectations w.r.t. it and, therefore, the multiple optimal stopping rule under the Accumulated Loss Policy is then obtained via direct application of Theorem \ref{thm:MultStop}.

\begin{thm}[Global Risk Transfer Case: ALP] \label{thm:ALP_new} Defining $W(t) = Z(t) -\widetilde{Z}(t)$, for $t=1,\hdots,T$ the multiple optimal stopping rule is given by (\ref{eq:tau}), where
\begin{align*}
	v^{1,1} &= \sum_{m=1}^{+\infty}   p_m \Big\{ \overline{F}_{S_m}(ALP)ALP + m\mu F_{GIG}(ALP; \ \lambda/\mu^2, m^2\lambda, 1/2) \Big\}, \\
	v^{L,1} &= \sum_{m=1}^{+\infty}   p_m \Big\{ \overline{F}_{S_m}(ALP)\max\{ALP, \ v^{L-1,1} \} + m\mu \big( F_{GIG}(ALP; \ \lambda/\mu^2, m^2\lambda, 1/2) \\
													&- F_{GIG}(v^{L-1,1}; \ \lambda/\mu^2, m^2\lambda, 1/2) \big) + v^{L-1,1} F_{S_m}(\min\{v^{L-1,1}, \ ALP \}) \Big\} + p_0 v^{L-1,1}, \\
	v^{L,l+1} &= \sum_{m=1}^{+\infty}   p_m \Big\{ \overline{F}_{S_m}(ALP)\max\{v^{L-1,l}+ALP, \ v^{L-1,l+1} \} \\
													 &+ v^{L-1,l} (F_{S_m}(ALP) - F_{S_m}(v^{L-1,l+1}-v^{L-1,l}) ) + m\mu \big( F_{GIG}(ALP; \ \lambda/\mu^2, m^2\lambda, 1/2) \\
													&- F_{GIG}(v^{L-1,l+1}-v^{L-1,l}; \ \lambda/\mu^2, m^2\lambda, 1/2) \big) + v^{L-1,l+1} F_{S_m}(\min\{v^{L-1,l+1}-v^{L-1,l}, \ ALP \}) \Big\} + p_0 v^{L-1,l+1}, \\
	v^{l,l} &= \sum_{m=1}^{+\infty}   p_m \Big\{ \overline{F}_{S_m}(ALP)ALP + m\mu F_{GIG}(ALP; \ \lambda/\mu^2, m^2\lambda, 1/2) \Big\}.
\end{align*}
\end{thm}

\subsection{Post Attachment Point Coverage (PAP)}
As in the ALP case, for the Post Attachment Point Coverage Policy we can also model the intra-annual losses before applying the insurance policy, but to calculate the cdf and pdf of both the insured process $\widetilde{Z}$ and the gain process $Z-\widetilde{Z}$ it will be necessary to consider an additional conditioning step. Since the insured loss is given by Definition \ref{def:PAP} it will be convenient to define the first time the aggregated loss process exceeds the threshold $PAP$, which can be formally defined as the following stopping time
\begin{equation}\label{eq:M_star}
M^*_m = \inf \left\{n \leq m : \sum_{k=1}^n X_k > PAP \right\}
\end{equation}
and, $M^*_m = +\infty,$ if $\sum_{k=1}^n X_k \leq PAP$.

\subsubsection{Local Risk Transfer Objective: Minimizing the Loss at the Stopping Times}
\begin{prop}[Local Risk Transfer Case: PAP]  \label{prop:PAP} The cdf and the pdf of the insured process are given, respectively, by
\begin{align*}
	F_{\widetilde{Z}}(z) &= \sum_{m=1}^{+\infty} \Bigg\{ \sum_{m^*=1}^m \Big( F_{IG}(z; \ m^*\mu, m^{*2}\lambda)D_{m^*,m} \Big) + F_{IG}(z; \ m\mu, \ m^2\lambda)D_m \Bigg\} +p_0 \\
	   f_{\widetilde{Z}}(z) &= \sum_{m=1}^{+\infty} \Bigg\{ \sum_{m^*=1}^m \Big( f_{IG}(z; \ m^*\mu, m^{*2}\lambda)D_{m^*,m} \Big) + f_{IG}(z; \ m\mu, \ m^2\lambda)D_m \Bigg\}\one_{z>0} +p_0\one_{z=0},
\end{align*}
where 
\begin{align*}
	 D_{m^*,m} &= \P\left[ M^*_m=m^*\right]p_m, \ m=1,2,\hdots, \ m^*=1,\hdots,m,\\
	       D_m &= F_{IG}(PAP, \ m\mu, \ m^2\lambda)p_m, \ \ m=1,2,\hdots
\end{align*}
\end{prop}

\begin{thm}[Local Risk Transfer Case: PAP] \label{thm:PAP} Defining $W(t)= -\widetilde{Z}(t)$, for $t=1,\hdots,T$ the multiple stopping rule is given by (\ref{eq:tau}), where
\begin{align*}
	v^{1,1} &=  -\sum_{m=1}^{+\infty} \sum_{m^*=1}^m m^*\mu D_{m^*,m}  + \sum_{m=1}^{+\infty} m \mu  D_m,\\
	v^{L,1} &= \sum_{m=1}^{+\infty} \sum_{m^*=1}^m \Bigg\{m\mu F_{GIG}(v^{L-1,1}; \ \lambda/\mu^2, m^{*2}\lambda, 1/2) +  v^{L-1,1} \overline{F}_{IG}(v^{L-1,1}; \ m^*\mu, m^{*2}\lambda) \Bigg\}D_{m^*,m} \\
					&+ \sum_{m=1}^{+\infty} \Big(m\mu F_{GIG}(v^{L-1,1}; \ \lambda/\mu^2, m^2\lambda, 1/2) +  v^{L-1,1} \overline{F}_{IG}(v^{L-1,1}; \ m\mu, m^2\lambda)\Big)D_m, \\	
v^{l,l+1} &= \sum_{m=1}^{+\infty} \sum_{m^*=1}^m \Bigg\{v^{L-1,l} F_{IG}(v^{L-1,l+1}-v^{L-1,l}; \ m^*\mu, m^{*2}\lambda)+ m\mu F_{GIG}(v^{L-1,l+1}-v^{L-1,l}; \ \lambda/\mu^2, m^{*2}\lambda, 1/2) \\
					&+  v^{L-1,l+1} \overline{F}_{IG}(v^{L-1,l+1}-v^{L-1,l}; \ m^*\mu, m^{*2}\lambda) \Bigg\}D_{m^*,m} \\
					&+ \sum_{m=1}^{+\infty} \Big(v^{L-1,l} F_{IG}(v^{L-1,l+1}-v^{L-1,l}; \ m\mu, m^2\lambda) + m\mu F_{GIG}(v^{L-1,l+1}-v^{L-1,l}; \ \lambda/\mu^2, m^2\lambda, 1/2) \\
					&+  v^{L-1,l+1} \overline{F}_{IG}(v^{L-1,l+1}-v^{L-1,l}; \ m\mu, m^2\lambda)\Big)D_m  + v^{L-1,l} p_0, \\
	v^{l,l} &= v^{l-1,l-1} -\sum_{m=1}^{+\infty} \sum_{m^*=1}^m m^*\mu D_{m^*,m}  + \sum_{m=1}^{+\infty} m \mu  D_m.
\end{align*}
\end{thm}

\subsubsection{Global Risk Transfer Objective: Minimizing the Loss Over Period $[0,T]$}
The gain process in the PAP--Total loss case takes the following form:
\begin{align*}
W &= Z - \widetilde{Z} \\
	&= \sum_{n=1}^N X_n -  \sum_{n=1}^{N} X_{n} \times \one_{\left\{ \sum_{k=1}^{n} X_k \leq PAP \right\}} \\
	&= \sum_{n=1}^N X_n \one_{\left\{ \sum_{k=1}^{n} X_k > PAP \right\}}.
\end{align*}

The resulting annual loss distribution and density are then given by Proposition \ref{prop:PAP_new}.
			
\begin{prop}[Global Risk Transfer Case: PAP]  \label{prop:PAP_new} The cdf and the pdf of the insured process are given, respectively, by
\begin{align*}
	F_W(w) &=  \sum_{m=1}^{+\infty} \Bigg\{ \sum_{m^*=1}^m \Big( \P\left[ \sum_{n=m^*}^m X_n \leq w \right]\P\left[ M^*_m=m^*\right] p_m\Big)\Bigg\} + \sum_{m=1}^{+\infty}\Bigg\{ P\left[ \sum_{k=1}^m X_k < PAP \right]p_m\Bigg\}  +p_0,\\
	   f_W(w) &= \left( \sum_{m=1}^{+\infty} \Bigg\{ \sum_{m^*=1}^m \Big( f_{IG}(w; \ (m-m^*+1)\mu, (m-m^*+1)^2\lambda)\P\left[ M^*_m=m^*\right] p_m\Big)\Bigg\} \right)\one_{\{ w> 0\}}\\
						&+ \P[W=0]\one_{\{ w= 0\}},
\end{align*}
where $P[W=0] = \sum_{m=1}^{+\infty}\Bigg\{ P\left[ \sum_{k=1}^m X_k < PAP \right]p_m\Bigg\}  +p_0$.
\end{prop}

\begin{thm}[Global Risk Transfer Case: PAP] \label{thm:PAP_new} Defining $W(t)= Z(t)-\widetilde{Z}(t)$, for $t=1,\hdots,T$ the multiple stopping rule is given by (\ref{eq:tau}), where
\begin{align*}
	v^{1,1} &=  \sum_{m=1}^{+\infty} \sum_{m^*=1}^m \P\left[ M^*_m=m^*\right] p_m (m-m^*+1)\mu,\\
	v^{L,1} &=  \sum_{m=1}^{+\infty} \sum_{m^*=1}^m \P\left[ M^*_m=m^*\right] p_m \Bigg\{ \overline{F}_{GIG}(v^{L-1,1}; \ \lambda/\mu^2, (m-m^*+1)^2\lambda, 1/2) (m-m^*+1)\mu, \\
														&+ v^{L-1,1} F_{GIG}(v^{L-1,1}; \ \lambda/\mu^2, (m-m^*+1)^2\lambda, -1/2) \Bigg\} + v^{L-1,1} \P[W=0]\\	
v^{l,l+1} &=  \sum_{m=1}^{+\infty} \sum_{m^*=1}^m \P\left[ M^*_m=m^*\right] p_m \Bigg\{ v^{L-1,l} \overline{F}_{GIG}(v^{L-1,l+1}-v^{L-1,l}; \ \lambda/\mu^2, (m-m^*+1)^2\lambda, -1/2) \\
														&+ \overline{F}_{GIG}(v^{L-1,l+1}-v^{L-1,l}; \ \lambda/\mu^2, (m-m^*+1)^2\lambda, 1/2) (m-m^*+1)\mu \\
														&+ v^{L-1,l+1} F_{GIG}(v^{L-1,l+1}-v^{L-1,l}; \ \lambda/\mu^2, (m-m^*+1)^2\lambda, -1/2) \Bigg\} + v^{L-1,l+1} \P[W=0],\\
	v^{l,l} &= v^{l-1,l-1}+ \sum_{m=1}^{+\infty} \sum_{m^*=1}^m \P\left[ M^*_m=m^*\right] p_m (m-m^*+1)\mu.
\end{align*}
\end{thm}

\subsection{Individual Loss Policy (ILP)}

The previous two insurance policies, the ALP and PAP structures, have been based on the aggregated amount throughout the year. In the case of the ILP insurance structure, the coverage is not on an accumulated aggregate coverage, instead it is based on an individual loss event coverage. 

\subsubsection{Local Risk Transfer Objective: Minimizing the Loss at the Stopping Times} \label{sec:ILP_loss1}
Let us assume a company buys the insurance policy called Individual Loss Policy (ILP). In this case, a particular loss process observed by the company \emph{after} applying the insurance policy may be given by
$$X_1(\omega)-TCL, \ X_2(\omega)-TCL, \ 0, \ 0, \ 0, \ X_6(\omega)-TCL, \ 0, \hdots, \ X_{N-1}(\omega)-TCL, \ 0.$$
In this case we can define a new process $(\widetilde{X}_n)_{n \geq 1}$ such that 
$$\widetilde{X}_1(\omega):=X_1(\omega) - TCL, \ \widetilde{X}_2(\omega):=X_2(\omega)-TCL, \ \widetilde{X}_3(\omega):=X_6(\omega)-TCL, \hdots,\widetilde{X}_{\widetilde{N}}(\omega):=X_{N-1}(\omega)-TCL$$
and the annual insured loss would be given by $\widetilde{Z} = \sum_{n=1}^{\widetilde{N}} \widetilde{X}_n$. Note that in this example the new process, $(\widetilde{X}_n)_{n \geq 1}$ would have $\widetilde{N} < N$ non zero observations and, in general, $\widetilde{N} \leq N$. The process $(\widetilde{X}_n)_{n \geq 1}$ can be interpreted as an auxiliary process, meaning that if the company had claimed on the insurance policy for this year then the observed losses would have been $\widetilde{X}_n$, instead of $X_n$.

In our approach we will model the random variable $\widetilde{N}$ and the process $(\widetilde{X}_n)_{n \geq 1}$, the first as an homogeneous Poisson process with mean $\widetilde{\lambda}_{\widetilde{N}}$ and the second as a sequence of i.i.d. random variables such that $\widetilde{X} \sim IG(\lambda, \ \mu)$.

\begin{thm}[Local Risk Transfer Case: ILP]  \label{thm:ILP} Assuming that $\widetilde{N} \sim Poi(\lambda_{\widetilde{N}})$ and $\widetilde{X}_1, \widetilde{X}_2,\hdots$ are i.i.d. with $\widetilde{X} \sim IG(\lambda, \ \mu)$ define $\widetilde{Z}(t) = \sum_{n=1}^{\widetilde{N}(t)} \widetilde{X}_n(t),$ and $W(t) = -\widetilde{Z}(t)$, for $t=1,\hdots,T$. In this case the optimal stopping rule is given by (\ref{eq:tau}), where
\begin{align*}
	v^{1,1} &= -\lambda_{\widetilde{N}} \mu, \\
	v^{L,1} &= -\sum_{n=1}^{+\infty} \mathrm{Pr}[\widetilde{N}=n] \Big[ F_{GIG}(v^{L-1,1}; \ \lambda/\mu^2, n^2\lambda, 1/2)n \mu -v^{L-1,1} F_{GIG}(v^{L-1,1}; \ \lambda/\mu^2, n^2\lambda, -1/2) + v^{L-1,1} \Big], \ \  1 < L \leq T, \\
	v^{L,l+1} &= -\sum_{n=1}^{+\infty} \mathrm{Pr}[\widetilde{N}=n] \Big[ F_{GIG}(v^{L-l,l+1}-v^{L-l,l}; \ \lambda/\mu^2, n^2\lambda, 1/2)n \mu \\
						&+ (v^{L-l,l}-v^{L-l,l+1}) F_{GIG}(v^{L-l,l+1}-v^{L-l,l}; \ \lambda/\mu^2, n^2\lambda, -1/2) + v^{L-l,l+1} \Big]+ v^{L-l,l}\mathrm{Pr}[\widetilde{N}=0], \ \ l+1 < L \leq T, \\
	v^{l,l} &= v^{l-1,l-1} -\lambda_{\widetilde{N}} \mu.
\end{align*}
\end{thm}

\subsubsection{Global Risk Transfer Objective: Minimizing the Loss Over Period $[0,T]$ via Monte Carlo}
If we assume the frequency of annual losses is given by $N \sim Poi(\lambda_N)$ and its severities by $X_i \sim IG(\lambda, \mu)$ then the gain process $W$ is given by
\begin{align*}
	W &= Z - \widetilde{Z} \\
	  &= \sum_{n=1}^N X_n  - \sum_{n=1}^{N} \max\left(X_{n} - \text{TCL}, \ 0\right) \\
		&= \sum_{n=1}^N X_n  - \sum_{n=1}^{N} \left(X_{n} - \text{TCL} \right)\one_{\{X_n > TCL \}} \\
		&= \sum_{n=1}^N \left( TCL \one_{\{ X_n > TCL \}} + X_n \one_{\{X_n \leq \}} \right) \\
		&= \sum_{n=1}^N  \min\{ X_n, \ TCL\}.
\end{align*}

From Lemma \ref{property1} we know the Inverse Gaussian family is closed under convolution, but the distribution of the sum of truncated Inverse Gaussian r.v.'s does not take any known form. A simple and effective way to approximate the expectations necessary to the calculation of the optimal multiple stopping rule is to use a Monte Carlo scheme as follows. \\
\begin{algorithm}[H]
 \SetAlgoLined
 \textbf{Inputs:} Model parameters $(\lambda, \mu, \lambda_N)$; Insurance policy parameter, $TCL$; Number of simulations $M$\;
 \KwResult{Simple Random Sample from the gain r.v. $W$ $(W^{(1)},\hdots, W^{(M)})$;}
 \For{$i=1,\hdots,M$}{
  Sample $N^{(i)} \sim Poi(\lambda_N)$\;
  \eIf{$N=0$}{
   $W^{(i)} = 0$\;
   }{
   Sample $X_k^{(i)} \sim IG(\lambda, \mu)$, for $k=1,\hdots,N^{(i)}$\;
   $W^{(i)} = \sum_{k=1}^{N^{(i)}} \min\{X_k^{(i)}, \ TCL \}$ \:
  }
 }
\end{algorithm}

By the end of this process we will have a sample $W^{(1)},\hdots,W^{(M)}$ from the gain, which can be used to approximate, for any given values of $0< c_1 < c_2 $ the expectations as
$$\E[ \max\{ c_1 + W, \ c_2 \}] \approx \frac{1}{M} \sum_{i=1}^M \max\{c_1 + W^{(i)}, \ c_2 \}.$$

\section{Case Studies} 
\label{sec:CaseStudies}
In this Section we will analyse the results provided by the optimal rule in the scenario where analytical expressions are available. Although the loss distribution parameters are different for each insurance policy, in this section we will assume the insurance product is valid for $T=8$ years and gives its owner the right to mitigate $k=3$ losses.

First, for the Accumulated Loss Policy (ALP), Figure \ref{fig:hist_ALP} presents a comparison of the two objective functions (Global and Local Risk Transfer), when the LDA parameters are $(\lambda, \mu, \lambda_N) = (3,2,3)$ and the insurance specific parameter is set to $ALP=10$. In this case we know the probability of having an annual loss that would make it worth utilising the insurance product in one year is $\P[Z> ALP] \approx 20\%$. In this study, for a large number of scenarios, $M=50,000$, the optimal rules from both the objective functions were calculated and, for each scenario, the set of stopping times $(m_1,m_2,m_3)$ were calculated. On the right-hand side of Figure \ref{fig:hist_ALP} we can see that the exercising strategy is considerably different for the two objective functions. For the Global Risk Transfer, we can see that fixing the first two stopping times, say $(m_1,m_2)=(1,2)$, it is preferable (on average) to use the remaining right as early as possible. Another way to see the same pattern is to verify that the frequency of occurrence of the set of strategies $(1,2,3); \ (2,3,4); \ (3,4,5); \ (4,5,6)$ is decreasing, again indicating a prevalence of early exercise strategies. On the other hand, if the objective is to minimize the ``local risk", in more than $50\%$ of the cases the optimal strategy will be to use the rights as soon as possible.

On the left-hand side of Figure \ref{fig:hist_ALP} we present histograms of the total loss over $[0,T]$ (i) without insurance (solid line); (ii) using the global objective function (dark grey); (iii) using the local objective function (light grey). As expected the mean of the total loss when using the local loss function is greater than the global one, but still smaller than the total loss without any insurance.

\begin{figure} 
	\vspace{-2.5cm}
	\subfloat{%
		\begin{minipage}[c][1\width]{%
			0.33\textwidth} 
			\centering%
			\includegraphics[width=1\textwidth]{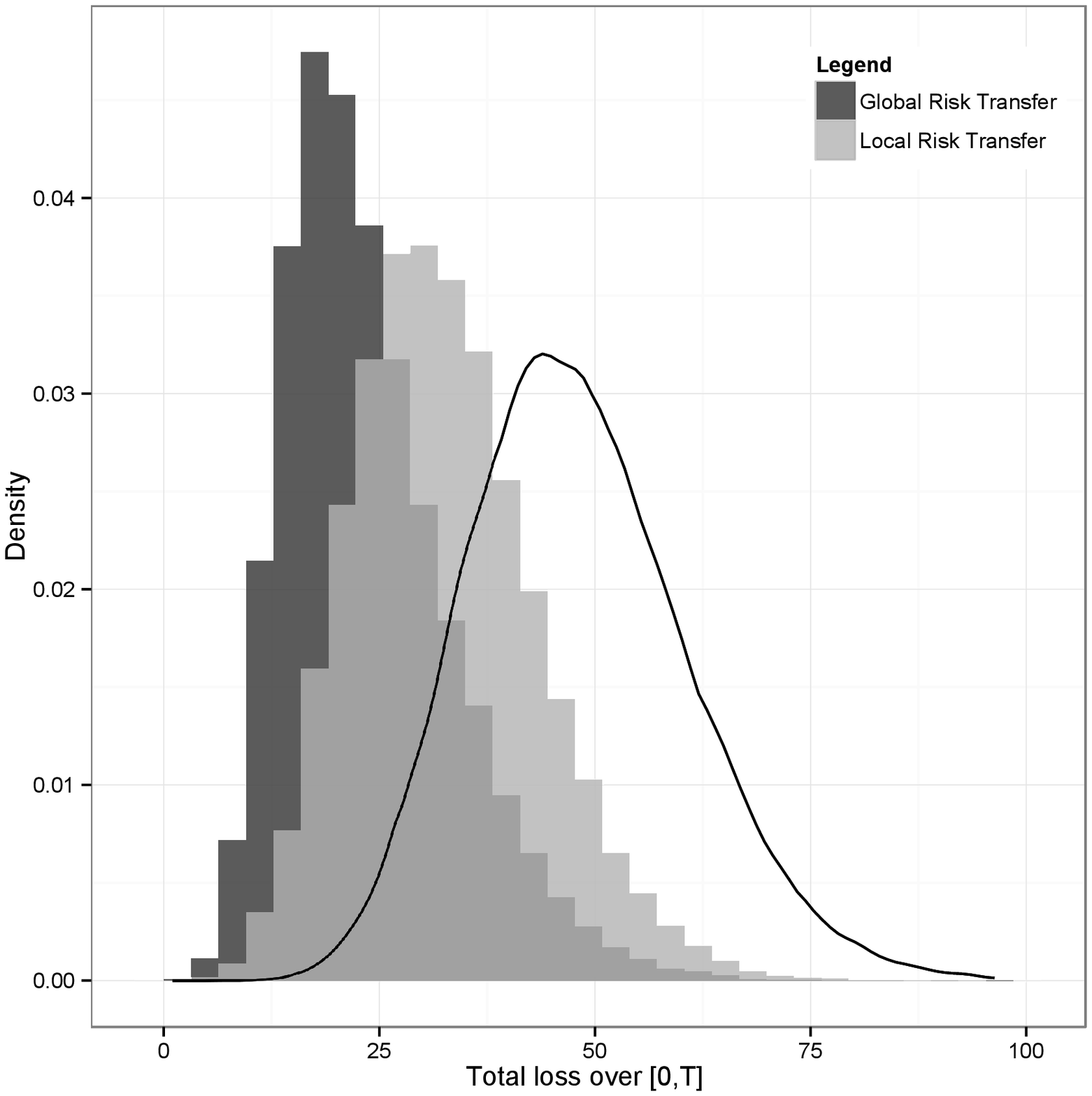}
		\end{minipage}
	} 
	\subfloat{%
		\begin{minipage}[c][1\width]{%
			0.67\textwidth} 
			\centering%
			\includegraphics[width=1\textwidth]{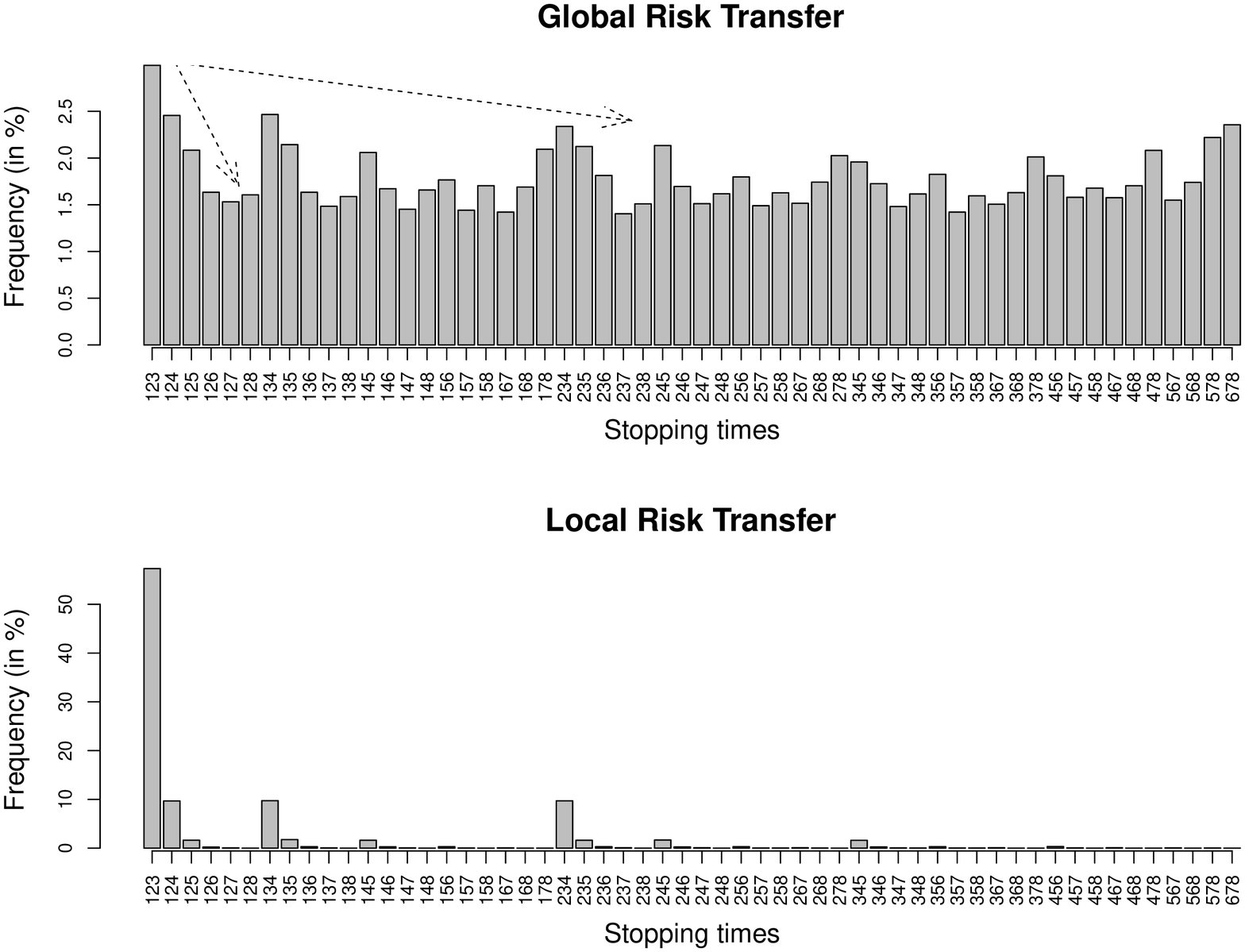}
		\end{minipage}
	} \vspace{-1.5cm}\caption{Comparison of the two objective functions using the Accumulated Loss Policy (ALP): (left) histograms of the total loss under the global objective function (dark grey), local objective function (light grey), no insurance case (solid line); (right) Multiple optimal stopping times under the two loss functions.} 
	\label{fig:hist_ALP}
\end{figure}

For both the PAP and the ILP case, we want to check the optimality of the rules presented, comparing them with pre-specified stopping rules. Denoting $(m_1,m_2,m_3)$ the three stopping times, the rules are defined as follows.
\begin{enumerate}[(i)]
	\item Rule 1 (Deterministic): Always stops at $m_1=1, m_2= 5, m_3=8$;
	\item Rule 2 (Random): Stops randomly at three points in $(1,\hdots,8)$, subject to $1 \leq m_1<m_2<m_3 \leq 8$;
	\item Rule 3 (Average): Stops when the observed loss is less then the expected loss, ie, $\E[W]$.
\end{enumerate}
For a large number of scenarios, $M=10,000$, we calculated the loss for each of the four rules (the Optimal, the Deterministic, the Random and the Average rules) and plot the histogram, comparing with the expected loss under the optimal rule, see Figure \ref{fig:hist_ALP} for the Accumulated Loss Policy (ALP) and Figure \ref{fig:hist_ILP} for the Individual Loss Policy (ILP). In all the examples the Optimal rule outperforms the other three showing the difficulty of creating a stopping rule that leads to losses as small as the optimal one.

In the first row of histograms on Figure \ref{fig:hist_PAP} the results are related to the global loss function, and in the second one to the local loss. Note that the horizontal axis in each line is exactly the objective function we are trying to minimize, precisely, $\displaystyle \sum_{t=1 \atop t\notin \{\tau_1,\hdots,\tau_k \}}^T \hspace{-0.5cm} Z(t) + \sum_{j=1}^k \widetilde{Z}(\tau_j)$ for the global optimization and $\sum_{j=1}^k \widetilde{Z}(\tau_j)$ for the local one. In this figure the vertical dashed bar represents the average total loss under the different rules and the solid grey line is defined as
\begin{enumerate}
	\item $\E[Z] \times T - v^{T,k}$, for the global optimization
	\item $v^{T,k}$, for the local optimization.
\end{enumerate}
These values must be understood as the expected loss under each of the two different gain functions and are easily derived from the definition of the gain functions and Theorem \ref{thm:MultStop}.

\begin{figure}[h!]
\centering
	\includegraphics[scale=0.6]{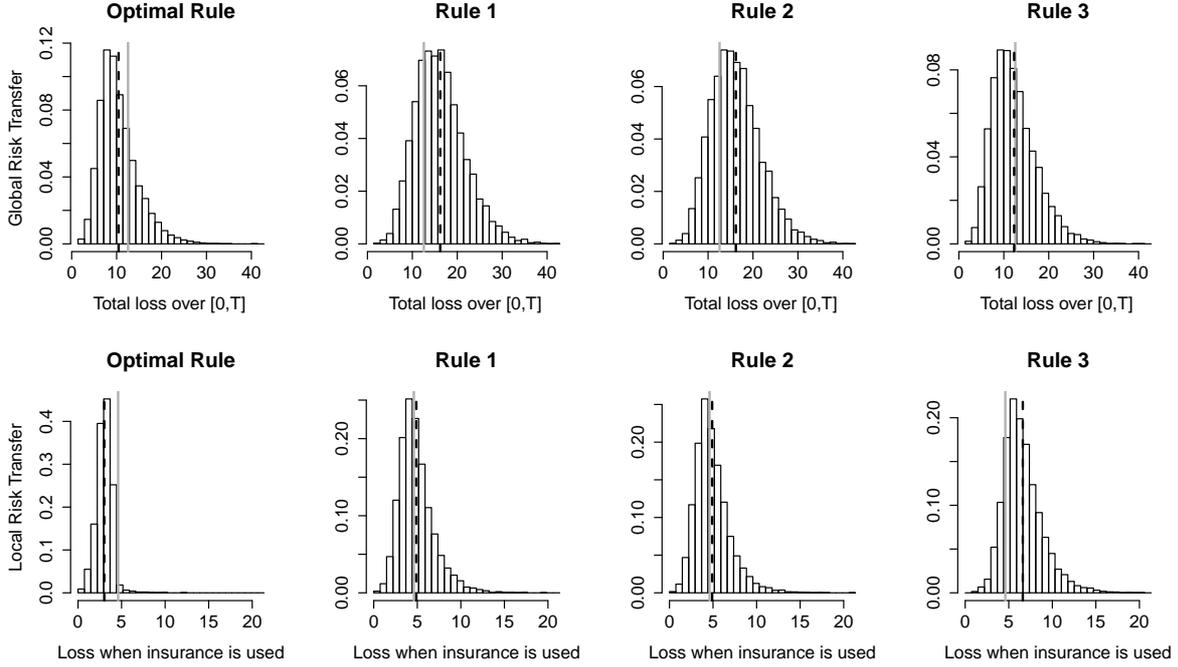}
	\caption{Histogram of losses under four different stopping rules for the PAP case with $(\lambda, \mu, \lambda_N) = (1,1,3)$.}
	\label{fig:hist_PAP}
\end{figure}

On Figure \ref{fig:hist_ILP} we present the same comparison as in the second row of Figure \ref{fig:hist_PAP} using the modelling proposed in Section \ref{sec:ILP_loss1}, with parameters $(\lambda, \mu, \lambda_{\widetilde{N}}) = (3,1,4)$. For this simulation study the conclusion is similar to the one drawn from the PAP case, where the pre-defined stopping rules underperformed the multiple optimal rule.

\begin{figure}[h!]
\centering
	\includegraphics[scale=0.6]{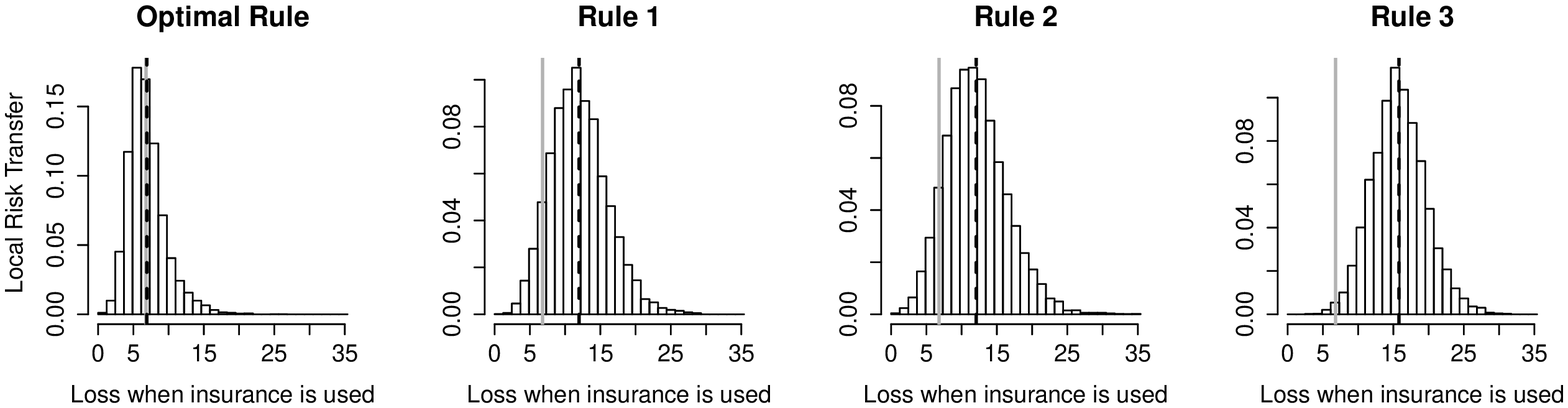}
	\caption{Histogram of losses under four different stopping rules for the ILP case with $(\lambda, \mu, \lambda_{\widetilde{N}}) = (3,1,4)$}
	\label{fig:hist_ILP}
\end{figure}

\section{Series expansion for the density of the insured process} 
\label{sec:SeriesExpansion}
Section \ref{section:AnalyticalOptimalRules} presented some combinations of Insurance Policies and LDA models that led to closed form solutions for the multiple stopping rule. For the cases where analytical solutions can not be found, one alternative is to create a series expansion of the density of the insured process $\widetilde{Z}$ such that all the expectations necessary in Theorem \ref{thm:MultStop} can be analytically calculated. In this Section we will assume the first $n$ moments of the distribution of the insured process $\widetilde{Z}$ are known and our objective is to minimize the local risk, but the calculations are also valid if we work with the global optimization problem (in this case one should use $Z-\widetilde{Z}$ instead of $\widetilde{Z}$).

\subsection{Gamma basis approximation}
If the $n$-th first moments of the insured process $\widetilde{Z}$ can be calculated (either algebraically or numerically) and the support of the insured random variable is $[0,+\infty)$ one can use a series expansion of the density of $\widetilde{Z}$ in a Gamma basis. For notational convenience, define a new random variable $U=b \widetilde{Z}$, where $b=\frac{\E[\widetilde{Z}]}{Var[\widetilde{Z}]}$ and set $a = \frac{\E[\widetilde{Z}]^2}{Var[\widetilde{Z}]}$. Denoting by $f_U$ the density of $U$ the idea, as in the Gaussian case of a Gram-Charlier expansion, is to write $f_U$ as
\begin{equation} \label{eq:f_U}
	f_U(u) = g(u;a)\left[A_0 L^{(a)}_0(u) + A_1  L^{(a)}_1(u) + A_2  L^{(a)}_2(u) + \hdots\right].
\end{equation}

Since $supp(U) = supp(\widetilde{Z})= [0,+\infty)$ we assume the kernel $g(\cdot \ ;a)$ also has positive support (differently from the Gram-Charlier expansion, where $g(\cdot)$ is chosen as a Gaussian kernel). If $g(u;a)= \frac{u^{a-1}e^{-u}}{\Gamma(a)}$ ie, a Gamma kernel with $shape=a$ and $scale =1$, then the orthonormal polynomial basis (with respect to this kernel) is given by the Laguerre polynomials (in contrast to Hermite polynomials in the Gaussian case) defined as
\begin{equation} \label{eq:Laguerre}
	L_n^{(a)}(u) = (-1)^n u^{1-a}e^{-u} \frac{d^n}{du^n}(u^{n+a-1}e^{-u}).
\end{equation}

\begin{table}[h!]
\centering
\begin{tabular}{lll}
\hline
	$L_0^{(a)}(u)$ =& $1$ \\
	$L_1^{(a)}(u)$ =& $u-a$ \\
	$L_2^{(a)}(u)$ =& $u^2 - 2(a+1)u + (a+1)a$\\
	$L_3^{(a)}(u)$ =& $u^3 - 3(a+2)u^2 + 3(a+2)(a+1)u - (a+2)(a+1)a$\\
	$L_4^{(a)}(u)$ =& $u^4 - 4(a+3)u^3 + 6(a+3)(a+2)u^2 - 4(a+3)(a+2)(a+1)u + (a+3)(a+2)(a+1)a$.\\
	\hline
\end{tabular}
\caption{The first five Laguerre polynomials}\label{tbl:Laguerre}
\end{table}

\begin{rem} Note that the definition of the Laguerre polynomials on Equation (\ref{eq:Laguerre}) is slightly different from the usual one, ie, the one based on Rodrigues' formula
$$\widetilde{L}_n^{(a)} = \frac{u^{-a}e^x}{n!}\frac{d^n}{du^n}\Big( e^{-x}x^{n+a}\Big),$$
but it is easy to check that
$$L_n^{(a)}(u) = n!(-1)^n\widetilde{L}_n^{(a-1)}.$$
\end{rem}

From the orthogonality condition (see, for example, \cite{jackson1941fourier} p. 184),
$$\displaystyle \int_0^{+\infty} \frac{x^{a-1}e^{-x}}{\Gamma(a)}L_n^{(a)}(x)L_m^{(a)}(x) dx = \begin{cases}\frac{n!\Gamma(a+n)}{\Gamma(a)}, \ n=m, \\ 0, \ n \neq m \end{cases}$$
and using the fact that $f_U$ can be written in the form of Equation (\ref{eq:f_U}) we find that 
\begin{equation} \label{eq:An}
	A_n=\frac{\Gamma(a)}{n!\Gamma(a+n)}\int_0^{+\infty} f_U(x)L_n^{(a)}(x)dx.
\end{equation}

Then, using the characterization of $A_n$ in (\ref{eq:An}) and the fact that $\E[U] = Var[U] = a$ we can see that
\begin{align*}
	A_0 &= \int_0^{+\infty} f_U(x) L_0^{(a)}(x)dx = \int_0^{+\infty} f_U(x) dx =1, \\
	A_1 &= \int_1^{+\infty} f_U(x) L_1^{(a)}(x)dx = \int_0^{+\infty} f_U(x) (z-a)dx =0, \\
	A_2 &= \int_1^{+\infty} f_U(x) L_2^{(a)}(x)dx = \int_0^{+\infty} f_U(x) (z^2 - 2(a+1)z + (a+1)a)dx =0.
\end{align*}
Similar but lengthier calculations show that for $\mu_n = \E\left[(U-\E[U])^n \right]$, $n=3,4$,
\begin{align} 
	A_3 &= \frac{\Gamma(a)}{3!\Gamma(a+3)}(\mu_3 - 2a), \label{eq:A3}\\
	A_4 &= \frac{\Gamma(a)}{4!\Gamma(a+4)}(\mu_4 - 12\mu_3 - 3a^2 +18a). \label{eq:A4}
\end{align}
Therefore, matching the first four moments, the density of the original random variable $\widetilde{Z}$ can be approximated as 
$$f_{\widetilde{Z}}(z)  = b f_U(u) \approx b \frac{u^{a-1}e^{-u}}{\Gamma(a)} \left[1 + A_3 L^{(a)}_3(u) + A_4  L^{(a)}_4(u) \right],$$
where $u=bz$, $A_3$ and $A_4$ are given, respectively, by (\ref{eq:A3}) and (\ref{eq:A4}) and the Laguerre polynomials can be found in Table \ref{tbl:Laguerre}. For additional details on the Gamma expansion we refer the reader to \cite{bowers1966expansion}.

Since this expansion does not ensure positivity of the density at all points (it can be negative for particular choices of skewness and kurtosis) we will adopt the approach discussed in \cite{jondeau2001gram} for the Gauss-Hermite Gramm Charlier case modified to the Gamma-Laguerre setting. To find the region on the $(\mu_3,\mu_4)$-plane where $f_U(u)$ is positive for all $u$ we will first find the region where $f_U(u) = 0$, i.e.,
\begin{equation} \label{eq:cond1_positivity}
	\frac{u^{a-1}e^{-u}}{\Gamma(a)}\big(1 + A_3 L^{(a)}_3(u) + A_4  L^{(a)}_4(u) \big)=0.
\end{equation}

For a fixed value $u$, we now want to find the set $(\mu_3,\mu_4)$ as a function of $u$ such that (\ref{eq:cond1_positivity}) remains zero for small variations on $u$. This set is given by $(\mu_3,\mu_4)$ such that
\begin{equation} \label{eq:cond2_positivity}
	\frac{d}{du}\left[\frac{u^{a-1}e^{-u}}{\Gamma(a)} \big(1 + A_3 L^{(a)}_3(u) + A_4  L^{(a)}_4(u)\big)\right] =0.
\end{equation}

We can then rewrite Equations (\ref{eq:cond1_positivity}) and (\ref{eq:cond1_positivity}) as the following system of algebraic equations
$$\left\{\begin{tabular}{ll}
	$\mu_3 B_1(u) + \mu_4 B_2(u)+ B_3(u) =0$\\
	$\mu_3 B_1'(u) + \mu_4 B_2'(u) + B_3'(u)  =0,$
				 \end{tabular} \right.$$
where
\begin{align*}
	B_1(u) &= \frac{u^{a-1}e^{-u}}{\Gamma(a)} \left(\frac{\Gamma(a)}{3!\Gamma(a+3)}L_3^{(a)}(u) -12\frac{\Gamma(a)}{4!\Gamma(a+4)}L_4^{(a)}(u) \right) ; \\
	B_2(u) &= \frac{u^{a-1}e^{-u}}{\Gamma(a)} \frac{\Gamma(a)}{4!\Gamma(a+4)}L_4^{(a)}(u) ; \\
	B_3(u) &= \frac{u^{a-1}e^{-u}}{\Gamma(a)}\left(1 - 2a\frac{\Gamma(a)}{3!\Gamma(a+3)}L_3^{(a)}(u) + \left(- 3a^2 +18a \right)\frac{\Gamma(a)}{4!\Gamma(a+4)}L_4^{(a)}(u) \right) ; \\
	B_1'(u) &= \left((a-1)u^{-1} -1 \right)B_1(u) + \frac{u^{a-1}e^{-u}}{\Gamma(a)} \left(\frac{\Gamma(a)}{3!\Gamma(a+3)}\frac{dL_3^{(a)}}{du}(u) -12\frac{\Gamma(a)}{4!\Gamma(a+4)}\frac{dL_4^{(a)}}{du}(u) \right); \\
	B_2'(u) &= \left((a-1)u^{-1} -1 \right)B_2(u) + \frac{u^{a-1}e^{-u}}{\Gamma(a)}\left(\frac{\Gamma(a)}{4!\Gamma(a+4)}\frac{dL_4^{(a)}}{du}(u)\right) ; \\
	B_3'(u) &= \left((a-1)u^{-1} -1 \right)B_3(u) + \frac{u^{a-1}e^{-u}}{\Gamma(a)} \left(-2a\frac{\Gamma(a)}{3!\Gamma(a+3)}\frac{dL_3^{(a)}}{du}(u) + \left(- 3a^2 +18a \right)\frac{\Gamma(a)}{4!\Gamma(a+4)}\frac{dL_4^{(a)}}{du}(u) \right) ; \\
	\frac{dL_3^{(a)}}{du}(u) &= 3u^2 - 6(a+2)u + 3(a+2)(a+1); \\
	\frac{dL_4^{(a)}}{du}(u) &= 4u^3 - 12(a+3)u^2 + 12(a+3)(a+2)u - 4(a+3)(a+2)(a+1).
\end{align*} 

Therefore, one can solve this system to show that the curve where the approximation will stay positive for all $u$ is given by:

\begin{equation} \label{eq:ApproximationRegion}
\left\{\begin{tabular}{ll}
	$\mu_4(u) = \left(\frac{B_1'B_3}{B_1} -B_3'\right)\left(B_2' - \frac{B_1'B_2}{B_1} \right)^{-1}$\\
	$\mu_3(u) = -\frac{1}{B_1}\left( \mu_4(u)B_2 + B_3 \right)$
				 \end{tabular} \right. , \text{ for } u\in[0,+\infty).
\end{equation}

As an illustration, Figure \ref{fig:ApproximationRegion} presents (on the left) the histogram of the loss process $Z=\sum_{n=1}^N X_n$ for $X \sim LN(\mu=1, \sigma=0.8)$ and $N\sim Poi(\lambda_N=2)$ and in red the Gamma approximation using the first four moments of $Z$. On the right it is presented the graph of the region where the density is positive for all values of $u$, given by equation \ref{eq:ApproximationRegion}. The grey area was calculated numerically, for all combinations in a fine grid on the plane $(\mu_3,\mu_4)$ it was tested if the density became negative in some point $z$. Grey points indicate the density is strictly positive. The blue point indicates the third and fourth moments in the Log-Normal example and since it lies inside the positivity are we can ensure this approximation is strictly positive for all values of $z$.

If the the third and fourth moments of the chosen model lied outside the permitted area one could chose $\widehat{\mu}_3$ and $\widehat{\mu}_4$ as the estimates that minimize some constrained optimization problem, for instance, the Maximum Likelihood Estimator (using $f_U(u; \mu_3, \mu_4) = \frac{u^{a-1}e^{-u}}{\Gamma(a)} \left[1 + A_3 L^{(a)}_3(u) + A_4  L^{(a)}_4(u)\right]$ as the likelihood). The constrained region is clearly given by a segment of the curve in equation \ref{eq:ApproximationRegion} and the endpoints can be found using a root-search method checking for which values of $u$ the red curve in Figure \ref{fig:ApproximationRegion} touches the grey area.

Given the approximation of $f_U$, and consequently of $f_{\widetilde{Z}}$, one can easily calculate the optimal multiple stopping rule, since $\E[\widetilde{Z}]$ is assumed to be known and $\E[\min\{ c_1 + \widetilde{Z}, \ c_2 \}]$ can be calculated as follows.

\begin{lemma} \label{lem:Gamma} If $G \sim Gamma(a,1)$, ie, $f_G(x) = \frac{x^{a-1}e^{-x}}{\Gamma(a)}$ then, similarly to Lemma \ref{property4} the following property holds
\begin{equation}
	x f_{G}(x;\ a, 1) \equiv af_G(x; \ a+1, 1).
\end{equation}
\end{lemma}

Using this notation we can rewrite the approximation of $\widetilde{Z}$ as 
\begin{align*}
	f_{\widetilde{Z}}(z)  &\approx f_G(bz; a,1)A^*_1 + f_G(bz; a+1,1)A^*_2 + f_G(bz; a+2,1)A^*_3 + f_G(bz; a+3,1)A^*_4+ f_G(bz; a+4,1)A^*_5,
\end{align*}
where $A^*_1=\left(1-\frac{\Gamma(a+3)}{\Gamma(a)}A_3+\frac{\Gamma(a+4)}{\Gamma(a)}A_4\right)b, \ A^*_2=\left(3\frac{\Gamma(a+3)}{\Gamma(a)}A_3-4\frac{\Gamma(a+4)}{\Gamma(a)}A_4\right)b, \ A^*_3=\left(-3\frac{\Gamma(a+3)}{\Gamma(a)}A_3+6\frac{\Gamma(a+4)}{\Gamma(a)}A_4\right)b,$ \ $A^*_4=\left(\frac{\Gamma(a+3)}{\Gamma(a)}A_3-4\frac{\Gamma(a+4)}{\Gamma(a)}A_4\right)b, A^*_5=\left(\frac{\Gamma(a+4)}{\Gamma(a)}A_4\right)b$.

Then, we can calculate the other main ingredient of Theorem \ref{thm:MultStop}, namely
\begin{align*}
\E[\min\{ c_1 + \widetilde{Z}, \ c_2 \}] &= \int_0^{+\infty} \min\{ c_1 + z, \ c_2 \}f_{\widetilde{Z}}(z) dz \\
															&= \int_0^{+\infty} \big( (c_1 + z)\one_{\{c_1+z<c_2\}} + c_2\one_{\{c_1+z \geq c_2\}}  \big)f_{\widetilde{Z}}(z) dz\\
															&= \int_0^{c_2-c_1} z f_{\widetilde{Z}}(z) dz + c_1 \int_0^{c_2-c_1} f_{\widetilde{Z}}(z) dz + c_2 \int_{c_2-c_1}^{+\infty} f_{\widetilde{Z}}(z) dz\\
															&= a\sum_{k=1}^5 F_G(b(c_2-c_1); a+k,1)A_k^* + c_1\sum_{k=1}^5 F_G(b(c_2-c_1); a-1+k,1)A_k^* \\
															&+ c_2\sum_{k=1}^5 \overline{F}_G(b(c_2-c_1); a-1+k,1)A_k^*.
\end{align*}

\section{Conclusion and Final Remarks} 
\label{sec:conclusion}
In this paper we studied some properties of an insurance product where its owner has the right to choose which of the next $k$ years the issuer should mitigate its annual losses. For three different forms of mitigation we presented as closed form solutions for the exercising strategy that minimized (on average) the sum of all annual losses in the next $T$ years. This model assumed a ``moderate tail" for the severity of the losses the owner incurs, namely a Poisson-Inverse Gaussian LDA model.

Although it is assumed the company already holds the proposed contract, the company can use the analysis presented on Figure \ref{fig:hist_ALP} as a proxy for the price of the insurance product. The value, from the company's point of view, of the insurance product should be the expected difference (under the natural probability) of the losses that would be incurred without the product and the losses incurred using the product in the most profitable way (for the buyer),
$$\displaystyle \E\left[\sum_{t=1}^L Z(t) - \left(\sum_{t=1 \atop t\notin \{\tau_1,\hdots,\tau_k \}}^L Z(t) + \sum_{j=1}^l \widetilde{Z}(\tau_j)\right)\right].$$
It must also be said this price does not include the premium asked by the insurance company and also does not take into consideration the fact that external insurance companies will not have access to the models used by the company but it can still be a valuable proxy.

An alternative to the results presented in Section \ref{sec:SeriesExpansion} can involve the use of a Monte Carlo method. If there exists a mechanism to sample from the severity distribution one can easily create a sample of the insured process $\widetilde{Z}$ and use this sample to calculate all the necessary expectations on Theorem \ref{thm:MultStop}. The advantage of this approach is that one can handle any combination of severity distribution and insurance policy, but it can be extremely time consuming and the variance of the estimative can be prohibitive. It is important to note the sampling of the severity can be made \emph{offline}, ie, the same sample should be used to calculate all the integrals. Another alternative to solve the optimal multiple stopping problem is the usage of an extended version of the so-called Least-Square Monte Carlo method, first presented in \cite{longstaff2001valuing}.

Regarding the results presented in Theorems \ref{thm:ILP} to \ref{thm:PAP} the truncation point for the infinite sums can be chosen to be much larger than the expected number of losses (parameter $\lambda_N$), since the summands are composed by a p.m.f. of a Poisson r.v. (which presents exponential decay) and a bounded term (difference of c.d.f.'s times constants).

\begin{figure}[h]
\centering
	\includegraphics[scale=0.4, angle=-90]{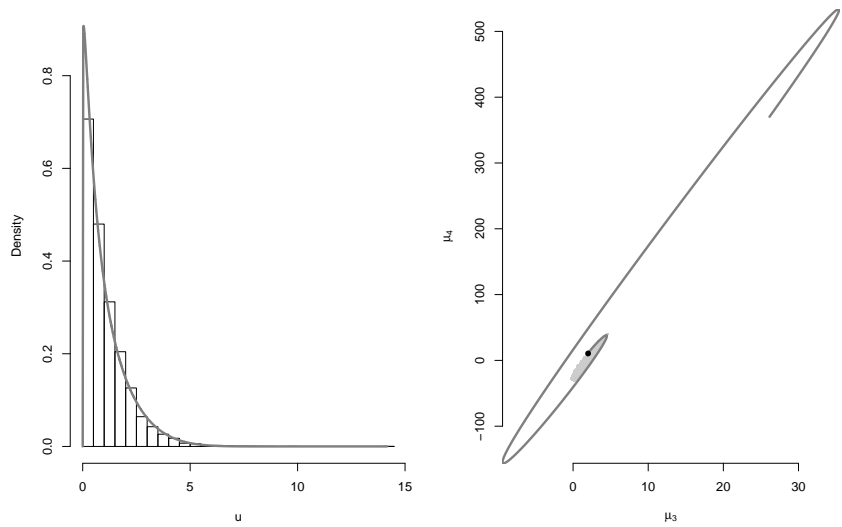}
	\caption{(Left) Histogram of the loss process $Z=\sum_{n=1}^N X_n$ for $X \sim LN(\mu=1, \sigma=0.8)$ and $N\sim Poi(\lambda_N=2)$ and in red the Gamma approximation using the first four moments of $Z$. (Right) The graph of the region where the density is positive for all values of $z$.}
	\label{fig:ApproximationRegion}
\end{figure}

\section*{Acknowledgements}
RST acknowledges the Conselho Nacional de Desenvolvimento Cient\'{i}ifico e Tecnol\'{o}gico (CNPq) for the Ci\^{e}ncia sem Fronteiras scholarship and the CSIRO Australia for support. GWP acknowledges the Insitute of Statistical Mathematics, Tokyo, Japan for support during research on this project.

%% file: fig_LDAmodel.tex
		\psscalebox{0.8 0.8} 
		{
		\begin{pspicture}(0,-3.5217187)(13.303333,3.5217187)
		\definecolor{colour0}{rgb}{0.7,0.7,0.7}
		\psline[linecolor=black, linewidth=0.02](3.7433333,-2.2514942)(3.7433333,-2.4114943)(3.7433333,-2.651494)
		\psline[linecolor=black, linewidth=0.02](6.943333,-2.2514942)(6.943333,-2.651494)(6.943333,-2.4914942)
		\psline[linecolor=black, linewidth=0.02](10.143333,-2.2514942)(10.143333,-2.651494)(10.143333,-2.651494)
		\rput[bl](11.943334,-3.0514941){Time}
		\rput[bl](3.5433333,-3.0514941){t=1}
		\rput[bl](6.7433333,-3.0514941){t=2}
		\rput[bl](9.943334,-3.0514941){t=3}
		\rput[bl](11.943334,-3.4514942){(in years)}
		\rput[bl](0.54333335,3.348506){Loss}
		\psframe[linecolor=black, linewidth=0.02, fillstyle=solid,fillcolor=colour0, dimen=outer](1.1433333,-1.4514941)(0.74333334,-2.4514942)
		\psframe[linecolor=black, linewidth=0.02, fillstyle=solid,fillcolor=colour0, dimen=outer](1.9433334,-0.45149413)(1.5433333,-2.4514942)
		\psframe[linecolor=black, linewidth=0.02, fillstyle=solid,fillcolor=colour0, dimen=outer](2.7433333,-1.4514941)(2.3433332,-2.4514942)
		\psframe[linecolor=black, linewidth=0.02, fillstyle=vlines*,fillcolor=colour0, hatchwidth=0.028222222, hatchangle=-45.0, hatchsep=0.1411111, dimen=outer](3.9433334,-1.4514941)(3.5433333,-2.4514942)
		\psframe[linecolor=black, linewidth=0.02, fillstyle=vlines*,fillcolor=colour0, hatchwidth=0.028222222, hatchangle=-45.0, hatchsep=0.1411111, dimen=outer](3.9433334,0.54850584)(3.5433333,-1.4514941)
		\psframe[linecolor=black, linewidth=0.02, fillstyle=vlines*,fillcolor=colour0, hatchwidth=0.028222222, hatchangle=-45.0, hatchsep=0.1411111, dimen=outer](3.9433334,1.5485059)(3.5433333,0.54850584)
		\psframe[linecolor=black, linewidth=0.02, fillstyle=solid,fillcolor=colour0, dimen=outer](4.943333,-1.4514941)(4.5433335,-2.4514942)
		\psframe[linecolor=black, linewidth=0.02, fillstyle=solid,fillcolor=colour0, dimen=outer](5.943333,1.5485059)(5.5433335,-2.4514942)
		\psframe[linecolor=black, linewidth=0.02, fillstyle=vlines*,fillcolor=colour0, hatchwidth=0.028222222, hatchangle=-45.0, hatchsep=0.1411111, dimen=outer](7.1433334,-1.4514941)(6.7433333,-2.4514942)
		\psframe[linecolor=black, linewidth=0.02, fillstyle=vlines*,fillcolor=colour0, hatchwidth=0.028222222, hatchangle=-45.0, hatchsep=0.1411111, dimen=outer](7.1433334,2.5485058)(6.7433333,-1.4514941)
		\psframe[linecolor=black, linewidth=0.02, fillstyle=solid,fillcolor=colour0, dimen=outer](9.143333,-1.4514941)(8.743334,-2.4514942)
		\psframe[linecolor=black, linewidth=0.02, fillstyle=solid,fillcolor=colour0, dimen=outer](8.143333,-0.5514941)(7.7433333,-2.4514942)
		\psframe[linecolor=black, linewidth=0.02, fillstyle=vlines*,fillcolor=colour0, hatchwidth=0.028222222, hatchangle=-45.0, hatchsep=0.1411111, dimen=outer](10.343333,0.44850585)(9.943334,-0.5514941)
		\psframe[linecolor=black, linewidth=0.02, fillstyle=vlines*,fillcolor=colour0, hatchwidth=0.028222222, hatchangle=-45.0, hatchsep=0.1411111, dimen=outer](10.343333,-0.5514941)(9.943334,-2.4514942)

		\rput(0.34333333,-2.4514942){\psaxes[linecolor=black, linewidth=0.02, tickstyle=full, axesstyle=axes, labels=y, ticks=y, ticksize=0.10583334cm, dx=1.0cm, dy=1.0cm]{->}(0,0)(0,0)(12,6)}

		\rput[bl](0.54333335,-1.3514942){$X_1(1)$}
		\rput[bl](1.3433334,-0.35149413){$X_2(1)$}
		\rput[bl](2.1433334,-1.3514942){$X_3(1)$}
		\rput[bl](3.2433333,1.7485058){$Z(1)$}
		\end{pspicture}
	}

%% file: fig_TCL.tex
		\psscalebox{0.8 0.8} 
		{
		\begin{pspicture}(0,-3.5217187)(13.293333,3.5217187)
		\definecolor{colour0}{rgb}{0.7,0.7,0.7}
		\psline[linecolor=black, linewidth=0.04](3.7333333,-2.2514942)(3.7333333,-2.4114943)(3.7333333,-2.651494)
		\psline[linecolor=black, linewidth=0.04](6.9333334,-2.2514942)(6.9333334,-2.651494)(6.9333334,-2.4914942)
		\psline[linecolor=black, linewidth=0.04](10.133333,-2.2514942)(10.133333,-2.651494)(10.133333,-2.651494)
		
		\rput[bl](11.933333,-3.0514941){Time}
		\rput[bl](3.5333333,-3.0514941){t=1}
		\rput[bl](6.733333,-3.0514941){t=2}
		\rput[bl](9.933333,-3.0514941){t=3}
		\rput[bl](11.933333,-3.4514942){(in years)}
		\rput[bl](0.53333336,3.348506){Loss}
		
		\psframe[linecolor=black, linewidth=0.02, fillstyle=solid,fillcolor=colour0, dimen=outer](1.9333334,-0.45149413)(1.5333333,-2.4514942)
		\psframe[linecolor=black, linewidth=0.02, fillstyle=vlines*, hatchwidth=0.028222222, hatchangle=-45.0, hatchsep=0.1411111, dimen=outer](3.9333334,-1.4514941)(3.5333333,-2.4514942)
		\psframe[linecolor=black, linewidth=0.02, fillstyle=vlines*, hatchwidth=0.028222222, hatchangle=-45.0, hatchsep=0.1411111, dimen=outer](3.9333334,0.54850584)(3.5333333,-1.4514941)
		\psframe[linecolor=black, linewidth=0.02, fillstyle=vlines*, hatchwidth=0.028222222, hatchangle=-45.0, hatchsep=0.1411111, dimen=outer](3.9333334,1.5485059)(3.5333333,0.54850584)
		\psframe[linecolor=black, linewidth=0.02, fillstyle=solid,fillcolor=colour0, dimen=outer](5.9333334,1.5485059)(5.5333333,-2.4514942)
		\psframe[linecolor=black, linewidth=0.02, fillstyle=vlines*, hatchwidth=0.028222222, hatchangle=-45.0, hatchsep=0.1411111, dimen=outer](7.133333,-1.4514941)(6.733333,-2.4514942)
		\psframe[linecolor=black, linewidth=0.02, fillstyle=vlines*, hatchwidth=0.028222222, hatchangle=-45.0, hatchsep=0.1411111, dimen=outer](7.133333,2.5485058)(6.733333,-1.4514941)
		\psframe[linecolor=black, linewidth=0.02, fillstyle=solid,fillcolor=colour0, dimen=outer](8.133333,-0.5514941)(7.733333,-2.4514942)
		\rput(0.33333334,-2.4514942){\psaxes[linecolor=black, linewidth=0.02, tickstyle=full, axesstyle=axes, labels=y, ticks=y, ticksize=0.10583334cm, dx=1.0cm, dy=1.0cm]{->}(0,0)(0,0)(12,6)}
		\rput[bl](1.3333334,-0.35149413){$X_2(1)$}
		\rput[bl](2.1333334,-1.3514942){$X_3(1)$}
		\rput[bl](3.2333333,1.7485058){$Z(1)$}
		\psline[linecolor=black, linewidth=0.02](0.33333334,-0.95149416)(11.833333,-0.95149416)(11.833333,-0.95149416)
		\rput[bl](11.333333,-0.75149417){TCL = 1.5}
		\psframe[linecolor=black, linewidth=0.02, fillstyle=solid, dimen=outer](1.1333333,-1.4514941)(0.73333335,-2.4514942)
		\psframe[linecolor=black, linewidth=0.02, fillstyle=solid, dimen=outer](4.9333334,-1.4514941)(4.5333333,-2.4514942)
		\psframe[linecolor=black, linewidth=0.02, fillstyle=solid, dimen=outer](2.7333333,-1.4514941)(2.3333333,-2.4514942)
		\psframe[linecolor=black, linewidth=0.02, fillstyle=solid, dimen=outer](1.9333334,-0.95149416)(1.5333333,-2.4514942)
		\psframe[linecolor=black, linewidth=0.02, fillstyle=solid, dimen=outer](8.133333,-0.95149416)(7.733333,-2.4514942)
		\psframe[linecolor=black, linewidth=0.02, fillstyle=solid, dimen=outer](5.9333334,-0.95149416)(5.5333333,-2.4514942)
		\psframe[linecolor=black, linewidth=0.02, fillstyle=solid, dimen=outer](9.133333,-1.4514941)(8.733334,-2.4514942)
		\rput[bl](0.53333336,-1.3514942){$X_1(1)$}
		\psframe[linecolor=black, linewidth=0.02, fillstyle=vlines*, hatchwidth=0.028222222, hatchangle=-45.0, hatchsep=0.1411111,fillcolor=colour0, dimen=outer](3.9333334,0.54850584)(3.5333333,0.048505858)
		\psframe[linecolor=black, linewidth=0.02, fillstyle=vlines*, hatchwidth=0.028222222, hatchangle=-45.0, hatchsep=0.1411111,fillcolor=colour0, dimen=outer](7.133333,2.5485058)(6.733333,0.048505858)
		\psframe[linecolor=black, linewidth=0.02, fillstyle=vlines*, hatchwidth=0.028222222, hatchangle=-45.0, hatchsep=0.1411111, dimen=outer](10.333333,0.44850585)(9.933333,-0.5514941)
		\psframe[linecolor=black, linewidth=0.02, fillstyle=vlines*, hatchwidth=0.028222222, hatchangle=-45.0, hatchsep=0.1411111, dimen=outer](10.333333,-0.5514941)(9.933333,-2.4514942)
		\psframe[linecolor=black, linewidth=0.02, fillstyle=vlines*, hatchwidth=0.028222222, hatchangle=-45.0, hatchsep=0.1411111,fillcolor=colour0, dimen=outer](10.333333,-0.5514941)(9.933333,-0.95149416)
		\end{pspicture}
		}

%% file: fig_ALP.tex
\psscalebox{0.8 0.8} 
{
\begin{pspicture}(0,-3.5217187)(13.293333,3.5217187)
\definecolor{colour0}{rgb}{0.7,0.7,0.7}
\psline[linecolor=black, linewidth=0.04](3.7333333,-2.2514942)(3.7333333,-2.4114943)(3.7333333,-2.651494)
\psline[linecolor=black, linewidth=0.04](6.9333334,-2.2514942)(6.9333334,-2.651494)(6.9333334,-2.4914942)
\psline[linecolor=black, linewidth=0.04](10.133333,-2.2514942)(10.133333,-2.651494)(10.133333,-2.651494)
\rput[bl](11.933333,-3.0514941){Time}
\rput[bl](3.5333333,-3.0514941){t=1}
\rput[bl](6.733333,-3.0514941){t=2}
\rput[bl](9.933333,-3.0514941){t=3}
\rput[bl](11.933333,-3.4514942){(in years)}
\rput[bl](0.53333336,3.348506){Loss}
\psframe[linecolor=black, linewidth=0.02, fillstyle=solid,fillcolor=colour0, dimen=outer](1.9333334,-0.45149413)(1.5333333,-2.4514942)
\psframe[linecolor=black, linewidth=0.02, fillstyle=solid,fillcolor=colour0, dimen=outer](2.7333333,-1.4514941)(2.3333333,-2.4514942)
\psframe[linecolor=black, linewidth=0.02, fillstyle=vlines*, hatchwidth=0.028222222, hatchangle=-45.0, hatchsep=0.1411111,fillcolor=colour0, dimen=outer](3.9333334,1.5485059)(3.5333333,0.54850584)
\psframe[linecolor=black, linewidth=0.02, fillstyle=solid,fillcolor=colour0, dimen=outer](5.9333334,1.5485059)(5.5333333,-2.4514942)
\psframe[linecolor=black, linewidth=0.02, fillstyle=vlines*, hatchwidth=0.028222222, hatchangle=-45.0, hatchsep=0.1411111,fillcolor=colour0, dimen=outer](7.133333,2.5485058)(6.733333,-1.4514941)
\psframe[linecolor=black, linewidth=0.02, fillstyle=solid,fillcolor=colour0, dimen=outer](9.133333,-1.4514941)(8.733334,-2.4514942)
\rput(0.33333334,-2.4514942){\psaxes[linecolor=black, linewidth=0.02, tickstyle=full, axesstyle=axes, labels=y, ticks=y, ticksize=0.10583334cm, dx=1.0cm, dy=1.0cm]{->}(0,0)(0,0)(12,6)}
\rput[bl](1.3333334,-0.35149413){$X_2(1)$}
\rput[bl](2.1333334,-1.3514942){$X_3(1)$}
\rput[bl](3.2333333,1.7485058){$Z(1)$}
\psline[linecolor=black, linewidth=0.02](0.33333334,-0.45149413)(11.833333,-0.45149413)(11.833333,-0.45149413)
\rput[bl](11.233334,-0.25149414){ALP = 2}
\psframe[linecolor=black, linewidth=0.02, fillstyle=solid, dimen=outer](1.1333333,-1.4514941)(0.73333335,-2.4514942)
\psframe[linecolor=black, linewidth=0.02, fillstyle=solid, dimen=outer](4.9333334,-1.4514941)(4.5333333,-2.4514942)
\psframe[linecolor=black, linewidth=0.02, fillstyle=solid, dimen=outer](1.9333334,-1.4514941)(1.5333333,-2.4514942)
\psframe[linecolor=black, linewidth=0.02, fillstyle=solid, dimen=outer](8.133333,-0.45149413)(7.733333,-2.4514942)
\psframe[linecolor=black, linewidth=0.02, fillstyle=solid, dimen=outer](5.9333334,-1.4514941)(5.5333333,-2.4514942)
\rput[bl](0.53333336,-1.3514942){$X_1(1)$}
\psframe[linecolor=black, linewidth=0.02, fillstyle=vlines*, hatchwidth=0.028222222, hatchangle=-45.0, hatchsep=0.1411111, dimen=outer](10.333333,-0.45149413)(9.933333,-2.4514942)
\psframe[linecolor=black, linewidth=0.02, fillstyle=vlines*, hatchwidth=0.028222222, hatchangle=-45.0, hatchsep=0.1411111,fillcolor=colour0, dimen=outer](3.9333334,0.54850584)(3.5333333,-0.45149413)
\psframe[linecolor=black, linewidth=0.02, fillstyle=vlines*, hatchwidth=0.028222222, hatchangle=-45.0, hatchsep=0.1411111, dimen=outer](7.133333,-1.4514941)(6.733333,-2.4514942)
\psframe[linecolor=black, linewidth=0.02, fillstyle=vlines*, hatchwidth=0.028222222, hatchangle=-45.0, hatchsep=0.1411111, dimen=outer](7.133333,-0.45149413)(6.733333,-1.4514941)
\psframe[linecolor=black, linewidth=0.02, fillstyle=vlines*, hatchwidth=0.028222222, hatchangle=-45.0, hatchsep=0.1411111, dimen=outer](3.9333334,-0.45149413)(3.5333333,-1.4514941)
\psframe[linecolor=black, linewidth=0.02, fillstyle=vlines*, hatchwidth=0.028222222, hatchangle=-45.0, hatchsep=0.1411111, dimen=outer](3.9333334,-1.4514941)(3.5333333,-2.4514942)
\psframe[linecolor=black, linewidth=0.02, fillstyle=vlines*, hatchwidth=0.028222222, hatchangle=-45.0, hatchsep=0.1411111,fillcolor=colour0, dimen=outer](10.333333,0.54850584)(9.933333,-0.45149413)
\end{pspicture}
}

%% file: fig_PAP.tex
\psscalebox{0.8 0.8} 
{
\begin{pspicture}(0,-3.5217187)(13.293333,3.5217187)
\definecolor{colour0}{rgb}{0.7,0.7,0.7}
\psline[linecolor=black, linewidth=0.04](3.7333333,-2.2514942)(3.7333333,-2.4114943)(3.7333333,-2.651494)
\psline[linecolor=black, linewidth=0.04](6.9333334,-2.2514942)(6.9333334,-2.651494)(6.9333334,-2.4914942)
\psline[linecolor=black, linewidth=0.04](10.133333,-2.2514942)(10.133333,-2.651494)(10.133333,-2.651494)
\rput[bl](11.933333,-3.0514941){Time}
\rput[bl](3.5333333,-3.0514941){t=1}
\rput[bl](6.733333,-3.0514941){t=2}
\rput[bl](9.933333,-3.0514941){t=3}
\rput[bl](11.933333,-3.4514942){(in years)}
\rput[bl](0.53333336,3.348506){Loss}
\psframe[linecolor=black, linewidth=0.02, fillstyle=solid,fillcolor=colour0, dimen=outer](1.1333333,-1.4514941)(0.73333335,-2.4514942)
\psframe[linecolor=black, linewidth=0.02, fillstyle=solid,fillcolor=colour0, dimen=outer](1.9333334,-0.45149413)(1.5333333,-2.4514942)
\psframe[linecolor=black, linewidth=0.02, fillstyle=vlines*, hatchwidth=0.028222222, hatchangle=-45.0, hatchsep=0.1411111,fillcolor=colour0, dimen=outer](3.9333334,0.54850584)(3.5333333,-1.4514941)
\psframe[linecolor=black, linewidth=0.02, fillstyle=solid,fillcolor=colour0, dimen=outer](4.9333334,-1.4514941)(4.5333333,-2.4514942)
\psframe[linecolor=black, linewidth=0.02, fillstyle=solid,fillcolor=colour0, dimen=outer](5.9333334,1.5485059)(5.5333333,-2.4514942)
\psframe[linecolor=black, linewidth=0.02, fillstyle=vlines*, hatchwidth=0.028222222, hatchangle=-45.0, hatchsep=0.1411111,fillcolor=colour0, dimen=outer](7.133333,-1.4514941)(6.733333,-2.4514942)
\psframe[linecolor=black, linewidth=0.02, fillstyle=solid,fillcolor=colour0, dimen=outer](9.133333,-1.4514941)(8.733334,-2.4514942)
\rput(0.33333334,-2.4514942){\psaxes[linecolor=black, linewidth=0.02, tickstyle=full, axesstyle=axes, labels=y, ticks=y, ticksize=0.10583334cm, dx=1.0cm, dy=1.0cm]{->}(0,0)(0,0)(12,6)}
\rput[bl](1.3333334,-0.35149413){$X_2(1)$}
\rput[bl](2.1333334,-1.3514942){$X_3(1)$}
\rput[bl](3.2333333,1.7485058){$Z(1)$}
\rput[bl](11.133333,0.74850583){PAP = 3}
\psframe[linecolor=black, linewidth=0.02, fillstyle=solid, dimen=outer](2.7333333,-1.4514941)(2.3333333,-2.4514942)
\psframe[linecolor=black, linewidth=0.02, fillstyle=solid, dimen=outer](8.133333,-0.45149413)(7.733333,-2.4514942)
\psframe[linecolor=black, linewidth=0.02, fillstyle=solid, dimen=outer](5.9333334,1.5485059)(5.5333333,-2.4514942)
\rput[bl](0.53333336,-1.3514942){$X_1(1)$}
\psframe[linecolor=black, linewidth=0.02, fillstyle=vlines*, hatchwidth=0.028222222, hatchangle=-45.0, hatchsep=0.1411111,fillcolor=colour0, dimen=outer](3.9333334,-1.4514941)(3.5333333,-2.4514942)
\psframe[linecolor=black, linewidth=0.02, fillstyle=vlines*, hatchwidth=0.028222222, hatchangle=-45.0, hatchsep=0.1411111, dimen=outer](7.133333,2.5485058)(6.733333,-1.4514941)
\psframe[linecolor=black, linewidth=0.02, fillstyle=vlines*, hatchwidth=0.028222222, hatchangle=-45.0, hatchsep=0.1411111, dimen=outer](3.9333334,1.5485059)(3.5333333,0.54850584)
\psframe[linecolor=black, linewidth=0.02, fillstyle=vlines*, hatchwidth=0.028222222, hatchangle=-45.0, hatchsep=0.1411111,fillcolor=colour0, dimen=outer](10.333333,0.54850584)(9.933333,-0.5514941)
\psframe[linecolor=black, linewidth=0.02, fillstyle=solid,fillcolor=colour0, dimen=outer](8.133333,-0.45149413)(7.733333,-2.4514942)
\psframe[linecolor=black, linewidth=0.02, fillstyle=vlines*, hatchwidth=0.028222222, hatchangle=-45.0, hatchsep=0.1411111,fillcolor=colour0, dimen=outer](10.333333,-0.45149413)(9.933333,-2.4514942)
\psline[linecolor=black, linewidth=0.02](0.33333334,0.54850584)(11.833333,0.54850584)(11.833333,0.54850584)

\end{pspicture}
}

%% file: appendix.tex
\appendix

\section{Proofs}

\begin{proof}(of Proposition \ref{prop:ALP}) As usual in LDA models, the first step to calculate the distribution of the compound process is to condition on the number of losses. Then the key point for this insurance policy is that if we knew how many losses occurred in a given year, say $m$ we would be able to check if $\sum_{n=1}^m X_n > ALP$. In this case the insured loss $\widetilde{Z}$ would be equal to $\sum_{n=1}^m X_n - ALP$. Otherwise all the losses would be insured and $\widetilde{Z}=0$. Using this argument we can calculate the cdf of $\widetilde{Z}$ as
\begin{align*}
	F_{\widetilde{Z}}(z) &= \P[\widetilde{Z} \leq z] \\
											 &= \sum_{m=1}^{+\infty} \P[\widetilde{S}_m \leq z] p_m + p_0 \\
											 &=\sum_{m=1}^{+\infty} \Bigg\{\P\left[\widetilde{S}_m \leq z \Big| \sum_{k=1}^m X_k > ALP\right] \P\left[\sum_{k=1}^m X_k > ALP \right]p_m \\
											&+ \P\left[\widetilde{S}_m \leq z \Big| \sum_{k=1}^m X_k \leq ALP\right]\P\left[\sum_{k=1}^m X_k \leq ALP \right]p_m \Bigg\}+ p_0 \\
											 &=\sum_{m=1}^{+\infty} \Bigg\{\P\left[\sum_{k=1}^m X_k - ALP \leq z \right] \P\left[\sum_{k=1}^m X_k > ALP \right]p_m \\
											&+ \P\left[0 \leq z \right]\P\left[\sum_{k=1}^m X_k \leq ALP \right]p_m \Bigg\}+ p_0 \\
											&= \sum_{m=1}^{+\infty} \Bigg\{\P\left[\sum_{k=1}^m X_k - ALP \leq z\right] \P\left[\sum_{k=1}^m X_k > ALP\right] p_m + \P\left[\sum_{k=1}^m X_k \leq ALP\right]p_m \Bigg\} + p_0 \\
											&= \sum_{m=1}^{+\infty} \Bigg\{F_{IG}(z+ALP; \ m\mu, m^2\lambda) \overline{F}_{IG}(ALP; \ m\mu, m^2\lambda) p_m + F_{IG}(ALP; \ m\mu, m^2\lambda)p_m \Bigg\} + p_0 \\
											&= \sum_{m=1}^{+\infty} F_{IG}(z+ALP; \ m\mu, m^2\lambda) \overline{F}_{IG}(ALP; \ m\mu, m^2\lambda) p_m + \underbrace{\sum_{m=1}^{+\infty}F_{IG}(ALP; \ m\mu, m^2\lambda)p_m  + p_0}_{\P[\widetilde{Z}=0]},
\end{align*}
where $\widetilde{S}_m = \left(\sum_{n=1}^m X_n - ALP \right) \one_{\left\{\sum_{n=1}^m X_n > ALP\right\}}.$

The p.d.f. easily follows from the derivation of $F_{\widetilde{Z}}(z)$ with respect to $z$ but it is important to note that $f_{\widetilde{Z}}$ is a continuous density with discrete mass at $z=0$, ie,
\begin{align*}
f_{\widetilde{Z}}(z) &= \sum_{m=1}^{+\infty} \Big\{f_{IG}(z+ALP; \ m\mu, m^2\lambda) \overline{F}_{IG}(ALP; \ m\mu, m^2\lambda) p_m \Big\}\one_{\{z>0\}} \\
									   &+ \Big\{p_0 + \sum_{m=1}^{+\infty}F_{IG}(ALP; \ m\mu, m^2\lambda)p_m \Big\}\one_{\{z=0\}}
\end{align*}
\end{proof}

\begin{proof}(of Theorem \ref{thm:ALP}) As in Theorem \ref{thm:ILP}, to calculate the optimal rule we only need to calculate $\E[W]$ and $\E[\max\{- c_1 + W, \ -c_2\}]$, for $0<c_1<c_2$. Given the expression (\ref{pdfALP}) for the density of $\widetilde{Z}$ we can calculate $\E[W]$ as follows
\begin{align*}
	\E[\widetilde{Z}] &= \int_0^{+\infty} z \sum_{m=1}^{+\infty} f_{IG}(z+ALP; \ m\mu, m^2\lambda)C_m dz \\
(\text{from Lemma }\ref{property2})&= \sum_{m=1}^{+\infty} C_m \int_0^{+\infty} z f_{GIG}(z+ALP; \ \lambda/\mu^2, m^2\lambda, -1/2) dz \\
(\text{change of variables})&= \sum_{m=1}^{+\infty} C_m \int_{ALP}^{+\infty} (w-ALP) f_{GIG}(w; \ \lambda/\mu^2, m^2\lambda, -1/2) \ dw \\
(\text{from Lemma } \ref{property4}) &= \sum_{m=1}^{+\infty} C_m \Big(m\mu \overline{F}_{GIG}(ALP; \ \lambda/\mu^2, m^2\lambda, 1/2) - ALP \overline{F}_{GIG}(ALP; \ \lambda/\mu^2, m^2\lambda, -1/2) \Big)
\end{align*}
And then we use the fact that $\E[W] = -\E[\widetilde{Z}]$. \\
For the second term we have that $\E[\max\{- c_1 + W, \ -c_2\}] = (-1)\E[\min\{ c_1 + \widetilde{Z}, \ c_2 \}]$ and 
\begin{align*}
	\E[\min\{ c_1 + \widetilde{Z}, \ c_2 \}] &= \int_0^{+\infty} \min\{ c_1 + z, \ c_2 \} f_{\widetilde{Z}}(z) dz \\
																					 &= \sum_{m=1}^{+\infty} C_m \int_0^{+\infty}\min\{ c_1 + z, \ c_2 \} f_{IG}(z+ALP; \ m\mu, m^2\lambda) dz \\
																					 &+\min\{ c_1 + 0, \ c_2 \} \Big\{ p_0 + \sum_{m=1}^{+\infty}F_{IG}(ALP; \ m\mu, m^2\lambda)p_m \Big\} \\
																					&= \sum_{m=1}^{+\infty} C_m \left[ \int_{ALP}^{+\infty} \min\{c_1+w-ALP, \ c_2 \} f_{IG}(w; \ m\mu, m^2\lambda) dw \right] + c_1 C_0 \\
																					&= \sum_{m=1}^{+\infty} C_m \left[ \int_{ALP}^{c_2-c_1+ALP} (c_1+w-ALP)f_{IG}(w; \ m\mu, m^2\lambda) dw + \int_{c_2-c_1+ALP}^{+\infty} c_2 f_{IG}(w; \ m\mu, m^2\lambda) dw \right] \\
																					&+ c_1 C_0 \\
																					&= \sum_{m=1}^{+\infty} C_m \Bigg[\Bigg(m\mu \Big(F_{GIG}(c_2-c_1+ALP; \ \lambda/\mu^2, m^2\lambda, 1/2)- F_{GIG}(ALP; \ \lambda/\mu^2, m^2\lambda, 1/2)\Big) \\
																	&+ (c_1-ALP) \Big(F_{GIG}(c_2-c_1+ALP; \ \lambda/\mu^2, m^2\lambda, -1/2)- F_{GIG}(ALP; \ \lambda/\mu^2, m^2\lambda, -1/2) \Big) \Bigg) \\
																	&+ c_2 \overline{F}_{GIG}(c_2-c_1+ALP; \ \lambda/\mu^2, m^2\lambda, -1/2) \Bigg] + c_1C_0.
\end{align*}
\end{proof}

\begin{proof}(of Proposition \ref{prop:ALP_new})
This proof follows from two conditioning arguments. The first part is to fix the number of annual losses $N=m$ and the second one is to separate the space where $\sum_{n=1}^m X_n > ALP$ and its complement. Formally,
\begin{align*}
\P[W\leq w] &= \sum_{m=1}^N \P[W_m\leq w \ | \ N=m ]\P[N=m] + \P[N=0],
\end{align*}
where
\begin{align*}
\P[W_m\leq w] &= \P\left[\min\left\{ALP, \ \sum_{n=1}^m X_n \right\} \leq w\right] \\
						&= \P\left[ALP  \leq w \ \Big| \ ALP \leq \sum_{n=1}^m X_n\right]\P\left[ALP \leq \sum_{n=1}^m X_n\right] \\
						&+  \P\left[\sum_{n=1}^m X_n  \leq w \ \Big| \ \sum_{n=1}^m X_n \leq ALP \right]\P\left[\sum_{n=1}^m X_n \leq ALP \right]\\
						&= \overline{F}_{S_m}(ALP) \one_{\left\{w \geq ALP \right\}} + \P\left[S_m \leq w, S_m \leq ALP \right] \\
						&= \overline{F}_{S_m}(ALP) \one_{\left\{w \geq ALP \right\}} + F_{S_m} (\min\{w, \ ALP\} ) \\
						&= \overline{F}_{S_m}(ALP) \one_{\left\{w \geq ALP \right\}} + F_{S_m} (ALP)\one_{\left\{w \geq ALP \right\}} + F_{S_m} (w)\one_{\left\{w < ALP \right\}} \\
						&= \one_{\left\{w \geq ALP \right\}} + F_{S_m} (w)\one_{\left\{w < ALP \right\}}.
\end{align*}
Consequently, the pdf of the gain is given by
$$f_W(w) = \sum_{m=1}^N\Big\{ \Big( \overline{F}_{S_m}(ALP)\one_{\left\{w = ALP \right\}} + f_{S_m} (w)\one_{\left\{0< w < ALP \right\}}   \Big)p_m \Big\}+ p_0\one_{\left\{w =0 \right\}}.$$

\end{proof}

\begin{proof}(Theorem \ref{thm:ALP_new})
For $0<c_1<c_2$, the quantity of interest can be calculated as
\begin{align*}
	\E[\max\{c_1+ W, c_2 \}] &= \int_0^{+\infty} \max\{c_1+ w, c_2 \} f_W(w) dw \\
													 &= \sum_{m=1}^{+\infty}  p_m \Big\{ \overline{F}_{S_m}(ALP)\max\{c_1+ALP, \ c_2 \} \\
													 &+ \int_{c_2-c_1}^{ALP}(c_1+w) f_{S_m} (w)dw + \int_0^{\min\{c_2-c_1, \ ALP \}}  c_2 f_{S_m} (w)dw\Big\} + p_0 c_2 \\
													 &= \sum_{m=1}^{+\infty}   p_m \Big\{ \overline{F}_{S_m}(ALP)\max\{c_1+ALP, \ c_2 \} \\
													 &+ c_1 (F_{S_m}(ALP) - F_{S_m}(c_2-c_1) ) + m\mu \big( F_{GIG}(ALP; \ \lambda/\mu^2, m^2\lambda, 1/2) \\
													&- F_{GIG}(c_2-c_1; \ \lambda/\mu^2, m^2\lambda, 1/2) \big) + c_2 F_{S_m}(\min\{c_2-c_1, \ ALP \}) \Big\} + p_0 c_2 \\
\end{align*}
\end{proof}
\begin{proof}(of Proposition \ref{prop:PAP}) This proof goes along the same lines as the Proof of Proposition \ref{prop:ALP}, but here we need one more conditioning step. This is due to the fact that we do not know in advance when the company will start to be covered by the insurance policy (this time is precisely the concept of the stopping time $M^*_m$ defined in (\ref{eq:M_star})). In the sequel we will denote $\widetilde{S}_m = X_{n} \times \one_{\left\{ \sum_{k=1}^{m} X_k \leq PAP \right\}}.$
	\begin{align*}
		F_{\widetilde{Z}}(z) &= \P[\widetilde{Z} \leq z] \\
											&= \sum_{m=1}^{+\infty} \P[\widetilde{S}_m \leq z]p_m + p_0 \\
											&= \sum_{m=1}^{+\infty} \Bigg\{ \sum_{m^*=1}^m \Big( \P\left[ \widetilde{S}_m \leq z \Big| M^*_m = m^*\right]\P\left[ M^*_m=m^*\right]p_m \Big) + \P\left[ \widetilde{S}_m \leq z \Big| M^*_m > m\right]\P\left[ M^*_m>m \right]p_m \Bigg\} +p_0 \\
											&= \sum_{m=1}^{+\infty} \Bigg\{ \sum_{m^*=1}^m \Big( \P\left[ \sum_{n=1}^{m^*} X_n \leq z \right]\P\left[ M^*_m=m^*\right]p_m \Big) + \P\left[ \sum_{n=1}^m X_n \leq z \right]\P\left[ \sum_{k=1}^m X_k < PAP \right]p_m \Bigg\} +p_0 \\
											&= \sum_{m=1}^{+\infty} \Bigg\{ \sum_{m^*=1}^m \Big( F_{IG}(z; \ m^*\mu, m^{*2}\lambda)\P\left[ M^*_m=m^*\right]p_m \Big) + F_{IG}(z; \ m\mu, \ m^2\lambda)F_{IG}(PAP, \ m\mu, \ m^2\lambda)p_m \Bigg\} +p_0 \\
											&= \sum_{m=1}^{+\infty} \Bigg\{ \sum_{m^*=1}^m \Big( F_{IG}(z; \ m^*\mu, m^{*2}\lambda)D_{m^*,m} \Big) + F_{IG}(z; \ m\mu, \ m^2\lambda)D_m \Bigg\} +p_0.
	\end{align*}
To calculate $D_{m^*,m}$, first define the (non insured) partial sum $S_m=\sum_{n=1}^m X_n$. Then, for $m^*=2,\hdots,m$
\begin{align*}
	\P\left[ M^*_m=m^*\right] &= \P\left[ S_{m^*} > PAP | S_{m^*-1} < PAP\right] \\
														&= \int_0^{PAP} \P\left[ S_{m^*} > PAP | S_{m^*-1} =a \right] f_{S_{m^*-1}}(a) da\\
														&= \int_0^{PAP} \P\left[ X_{m^*} > PAP -a \right] f_{S_{m^*-1}}(a) da\\
														&= \int_0^{PAP} \left[ 1 - F_{IG}(PAP -a; \ \mu, \ \lambda)\right] f_{IG}(a; (m^*-1)\mu, \ (m^*-1)^2\lambda) da
\end{align*}
and $\P\left[ M^*_m=1\right] = \P[X_1>PAP]$. It is important to note that no matter how many losses are in a year the probability of the sum of the first $m^*$ losses exceed the threshold $PAP$ is the same. Mathematically, it is equivalent to say that $\P\left[ M^*_m=m^*\right]$ does not depend on $m$. Another important aspect is that this integral can not be solved analytically but due to the fact that it is a uni-dimensional integral of well behaved integrands in a bounded set it can be easily approximated by any quadrature rule.
\end{proof}

\begin{proof}(of Theorem \ref{thm:PAP}) Remember that $\E[W] = - \E[\widetilde{Z}]$ and the former can be calculated as follows
\begin{align*}
	\E[\widetilde{Z}] &= \int_0^{+\infty} z f_{\widetilde{Z}}(z)dz \\
										&= \sum_{m=1}^{+\infty} \sum_{m^*=1}^m \int_0^{+\infty} z f_{IG}(z; \ m^*\mu, m^{*2}\lambda)D_{m^*,m}dz  + \sum_{m=1}^{+\infty} \int_0^{+\infty} zf_{IG}(z; \ m\mu, \ m^2\lambda)  D_m dz\\
										&= \sum_{m=1}^{+\infty} \sum_{m^*=1}^m m^*\mu D_{m^*,m}  + \sum_{m=1}^{+\infty} m \mu  D_m.
\end{align*}

For $0<c_1<c_2$ we have that
\begin{align*}
		\E[\min\{c_1+\widetilde{Z}, \ c_2 \}] &= \sum_{m=1}^{+\infty} \sum_{m^*=1}^m \int_0^{+\infty} \min\{c_1+z, \ c_2 \} f_{IG}(z; \ m^*\mu, m^{*2}\lambda)D_{m^*,m}dz \\
																					&+ \sum_{m=1}^{+\infty} \int_0^{+\infty} \min\{c_1+z, \ c_2 \}f_{IG}(z; \ m\mu, m^2\lambda)  D_m dz + \min\{c_1+0, \ c_2 \}p_0 \\
																					&= \sum_{m=1}^{+\infty} \sum_{m^*=1}^m \Big(\int_0^{c_2-c_1} (c_1+z)f_{IG}(z; \ m^*\mu, m^{*2}\lambda)dz +  \int_{c_2-c_1}^{+\infty} c_2 f_{IG}(z; \ m\mu, m^2\lambda)dz\Big)D_{m^*,m}dw \\
																					&+ \sum_{m=1}^{+\infty} \Big(\int_0^{c_2-c_1} (c_1+z)f_{IG}(z; \ m\mu, m^2\lambda)dz +  \int_{c_2-c_1}^{+\infty} c_2 f_{IG}(z; \ m\mu, m^2\lambda)dz\Big)D_m dz + c_1 p_0 \\
																					&= \sum_{m=1}^{+\infty} \sum_{m^*=1}^m \Bigg\{c_1 F_{IG}(c_2-c_1; \ m^*\mu, m^{*2}\lambda)+ m\mu F_{GIG}(c_2-c_1; \ \lambda/\mu^2, m^{*2}\lambda, 1/2) \\
																					&+  c_2 \overline{F}_{IG}(c_2-c_1; \ m^*\mu, m^{*2}\lambda) \Bigg\}D_{m^*,m} \\
																					&+ \sum_{m=1}^{+\infty} \Big(c_1 F_{IG}(c_2-c_1; \ m\mu, m^2\lambda) + m\mu F_{GIG}(c_2-c_1; \ \lambda/\mu^2, m^2\lambda, 1/2) \\
																					&+  c_2 \overline{F}_{IG}(c_2-c_1; \ m\mu, m^2\lambda)\Big)D_m  + c_1 p_0.
\end{align*}
The result follows from the equality $\E[\max\{- c_1 + W, \ -c_2\}] = (-1)\E[\min\{ c_1 + \widetilde{Z}, \ c_2 \}]$ and an application of Theorem \ref{thm:MultStop}.

\end{proof}

\begin{proof}(Proposition \ref{prop:PAP_new})
For this proof we first define $W_m$ as the gain process conditional on $N=m$, i.e.,
			$$W_m = \sum_{n=1}^{m} X_{n} \times \one_{\left\{ \sum_{k=1}^n X_k > PAP \right\}}.$$

Then, we can see that
	\begin{align*}
	\P[ W_m \leq w ] &= \sum_{m^*=1}^m \Big( \P\left[ W_m \leq z \Big| M^*_m = m^*\right]\P\left[ M^*_m=m^*\right]\Big) + \P\left[ W_m \leq z \Big| M^*_m = +\infty\right]\P\left[ M^*_m=+\infty\right]\\
	                             &= \sum_{m^*=1}^m \Big( \P\left[ \sum_{n=m^*}^m X_n \leq z \right]\P\left[ M^*_m=m^*\right] \Big) + \P\left[ \sum_{k=1}^m X_k < PAP \right].
	\end{align*}

Now, given the conditional distribution of the gain we simply need to weight each term by the probability of $N=m$ annual losses to calculate both the cdf and the pdf of the gain:
	\begin{align*}
		F_{\widetilde{Z}}(z) &= \P[\widetilde{Z} \leq z] \\
											&= \sum_{m=1}^{+\infty} \P[W_m \leq z]p_m + p_0 \\
											&= \sum_{m=1}^{+\infty} \Bigg\{ \sum_{m^*=1}^m \Big( \P\left[ \sum_{n=m^*}^m X_n \leq z \right]\P\left[ M^*_m=m^*\right] p_m\Big) + \P\left[ \sum_{k=1}^m X_k < PAP \right]p_m \Bigg\} +p_0 \\
											&= \sum_{m=1}^{+\infty} \Bigg\{ \sum_{m^*=1}^m \Big( \P\left[ \sum_{n=m^*}^m X_n \leq z \right]\P\left[ M^*_m=m^*\right] p_m\Big)\Bigg\} + \underbrace{\sum_{m=1}^{+\infty}\Bigg\{ P\left[ \sum_{k=1}^m X_k < PAP \right]p_m\Bigg\}  +p_0}_{\P[W=0]}, \\
	f_{\widetilde{Z}}(z) &= \left( \sum_{m=1}^{+\infty} \Bigg\{ \sum_{m^*=1}^m \Big( f_{IG}(z; \ (m-m^*+1)\mu, (m-m^*+1)^2\lambda)\P\left[ M^*_m=m^*\right] p_m\Big)\Bigg\} \right)\one_{\{ w> 0\}}\\
											&+ \P[W=0]\one_{\{ w= 0\}}.
	\end{align*}
	
\end{proof}

\begin{proof}(of Theorem \ref{thm:PAP_new})
Given the pdf of the gain the calculation of the necessary expectation is straightforward:
	\begin{align*}
		\E[\max\{c_1+ W, c_2 \}] &= \int_0^{+\infty} \max\{c_1+ w, c_2 \} f_W(w) dw \\
														 &= \sum_{m=1}^{+\infty} \sum_{m^*=1}^m \P\left[ M^*_m=m^*\right] p_m \Bigg\{ c_1 \overline{F}_{GIG}(c_2-c_1; \ \lambda/\mu^2, (m-m^*+1)^2\lambda, -1/2) \\
														&+ \overline{F}_{GIG}(c_2-c_1; \ \lambda/\mu^2, (m-m^*+1)^2\lambda, 1/2) (m-m^*+1)\mu \\
														&+ c_2 F_{GIG}(c_2-c_1; \ \lambda/\mu^2, (m-m^*+1)^2\lambda, -1/2) \Bigg\} + c_2 \P[W=0]
	\end{align*}
\end{proof}

\begin{proof}(of Theorem \ref{thm:ILP})
It is clear from Theorem \ref{thm:MultStop} that we only need to calculate two terms, namely $\E[W]$ and $\E[\max\{- c_1 + W, \ -c_2\}]$, for $0<c_1<c_2$. The first term can be derived by a simple application of the Tower Property:
$$\E[W] = -\E[\widetilde{Z}] = -\E\Big[ \E[\widetilde{Z} | \widetilde{N}] \Big] = -\E[\widetilde{N}] \ \E[\widetilde{X}] = -\lambda_{\widetilde{N}} \mu.$$


For the second term, first note that $\E[\max\{- c_1 + W, \ -c_2\}] = (-1)\E[\min\{ c_1 + \widetilde{Z}, \ c_2 \}]$ and it then follows that, for $0<c_1<c_2$,
\begin{align*}
\E[\min\{ c_1 + \widetilde{Z}, \ c_2 \}] &= \int_0^{+\infty} \min\{ c_1 + z, \ c_2 \}f_{\widetilde{Z}}(z) dz + \min\{c_1 + 0, \ c_2 \}\mathrm{Pr}[\widetilde{N}=0] \\
															&= \int_0^{+\infty} \big( (c_1 + z)\one_{\{c_1+z<c_2\}} + c_2\one_{\{c_1+z \geq c_2\}}  \big)f_{\widetilde{Z}}(z) dz + c_1\mathrm{Pr}[\widetilde{N}=0] \\
															&= \int_0^{c_2-c_1} z f_{\widetilde{Z}}(z) dz + c_1 \int_0^{c_2-c_1} f_{\widetilde{Z}}(z) dz + c_2 \int_{c_2-c_1}^{+\infty} f_{\widetilde{Z}}(z) dz + c_1\mathrm{Pr}[\widetilde{N}=0] \\
															&= \sum_{n=1}^{+\infty} \mathrm{Pr}[\widetilde{N}=n] \Big[ \int_0^{c_2-c_1} z f_{\widetilde{S}_n}(z) dz + c_1 \int_0^{c_2-c_1} f_{\widetilde{S}_n}(z) dz + c_2 \int_{c_2-c_1}^{+\infty} f_{\widetilde{S}_n}(z) dz \Big] \\
															&+ c_1\mathrm{Pr}[\widetilde{N}=0] \\
  (\text{from Lemma } \ref{property2})&= \sum_{n=1}^{+\infty} \mathrm{Pr}[\widetilde{N}=n] \Big[ F_{GIG}(c_2-c_1; \ \lambda/\mu^2, n^2\lambda, 1/2)n \mu + c_1 F_{GIG}(c_2-c_1; \ \lambda/\mu^2, n^2\lambda, -1/2) \\
															&+ c_2 \overline{F}_{GIG}(c_2-c_1; \ \lambda/\mu^2, n^2\lambda, -1/2) \Big]+ c_1\mathrm{Pr}[\widetilde{N}=0] \\
  (\text{from Lemma }\ref{property4})&= \sum_{n=1}^{+\infty} \mathrm{Pr}[\widetilde{N}=n] \Big[ F_{GIG}(c_2-c_1; \ \lambda/\mu^2, n^2\lambda, 1/2)n \mu + (c_1-c_2) F_{GIG}(c_2-c_1; \ \lambda/\mu^2, n^2\lambda, -1/2) + c_2 \Big] \\
															&+ c_1\mathrm{Pr}[\widetilde{N}=0].
\end{align*}
Note that, for notational ease, $f_{\widetilde{Z}}$ must be understood as the absolutely continuous part of the density of $\widetilde{Z}$.
\end{proof}